\documentclass{article}
\usepackage{arxiv}

\usepackage{authblk}
\usepackage{hyperref}

\usepackage[american]{babel}

\usepackage{booktabs} 

\usepackage{natbib} 
\usepackage{mathtools} 
\usepackage{tikz} 
\usepackage{pgfplots}
\usetikzlibrary{arrows.meta}
\usetikzlibrary{calc}

\usepackage{amsmath,amsfonts,amssymb,dsfont,bm,amsthm}

\usepackage{csquotes}
\usepackage{mathbbol}
\usepackage{subcaption}
\usepackage{thmtools,thm-restate}
\usepackage{MnSymbol}
\usepackage{bm}
\usepackage{accents}
\usepackage{hyperref}

\usepackage{enumitem}

\usepackage{ciphod}

\newtheorem{definition}{Definition}

\newtheorem{assumption}{Assumption}

\newtheorem{problem}{Problem}

\title{On Efficient Adjustment for Micro Causal Effects in Summary Causal Graphs}

\author[1]{Isabela~Belciug}

\author[1]{Simon~Ferreira}

\author[1]{Charles~K.~Assaad}

\affil[1]{Sorbonne Université, INSERM, Institut Pierre Louis d’Epidémiologie et de Santé Publique, F75012, Paris, France}

\date{}

\begin{document}
\maketitle

\begin{abstract}
Observational studies in fields such as epidemiology often rely on covariate adjustment to estimate causal effects. Classical graphical criteria, like the back-door criterion and the generalized adjustment criterion, are powerful tools for identifying valid adjustment sets in directed acyclic graphs (DAGs). However, these criteria are not directly applicable to summary causal graphs (SCGs), which are abstractions of DAGs commonly used in dynamic systems. In SCGs, each node typically represents an entire time series and may involve cycles, making classical criteria  inapplicable for identifying causal effects.
Recent work established complete conditions for determining whether the micro causal effect of a treatment or an exposure $X_{t-\gamma}$ on an outcome $Y_t$ is identifiable via covariate adjustment in SCGs, under the assumption of no hidden confounding. However, these identifiability conditions have two main limitations. First, they are complex, relying on cumbersome definitions and requiring the enumeration of multiple paths in the SCG, which can be computationally expensive. Second, when these conditions are satisfied, they only provide two valid adjustment sets, limiting flexibility in practical applications.
In this paper, we propose an equivalent but simpler formulation of those identifiability conditions and introduce a new criterion that identifies a broader class of valid adjustment sets in SCGs. Additionally, we characterize the quasi-optimal adjustment set among these, i.e., the one that minimizes the asymptotic variance of the causal effect estimator. Our contributions offer both theoretical advancement and practical tools for more flexible and efficient causal inference in abstracted causal graphs.
\end{abstract}

\section{Introduction}
Understanding causal relationships is a central goal in many scientific disciplines, such as epidemiology~\citep{Hernan_2023,Glemain_2025}, economics~\citep{Hunermund_2023}, climate science~\cite{Runge_2019}, and industry~\citep{Assaad_2023}. Causal inference from observational data, in particular, poses fundamental challenges due to the presence of confounding bias. 
A widely used framework for representing and reasoning about causal relationships is based on Structural Causal Models (SCM)~\citep{Pearl_2000}. An SCM is a set of structural equations that define how each observed variable is generated from its causes, thereby specifying a causal data-generating process. 
Usually, the SCM is supposed to be unknown, but the Directed Acyclic Graph (DAG) it induces\textemdash representing the qualitative causal structure among variables\textemdash is known.
In the induced DAG, nodes represent variables and directed edges denote direct causal influences. These graphs provide a powerful language for articulating causal assumptions and determining whether a given causal effect is identifiable from observational data. 
There exist several tools for identifying (i.e., removing confounding bias without adding new biases) a causal effect of interest in DAGs among which the do-calculus\textemdash a set of three rules that, when applied iteratively,\textemdash form a sound and complete tool for causal effect identification in non-parametric settings. However, while powerful in theory, the do-calculus~\citep{Pearl_2000} can be highly complex to apply in practice today. Even when identification is possible, the resulting formulas (e.g., the front-door formula~\citep{Pearl_2000} ~\citep{Piccininni_2023}) may not be compatible with standard estimators such as ordinary least squares regression or propensity score methods. As a result, researchers often prefer simpler graphical criteria like the back-door criterion~\citep{Pearl_2000} or the generalized adjustment criterion~\citep{Perkovic_2016}, which offer intuitive graphical tests to identify  causal effects by covariate adjustment which is compatible with standard estimators.

However, in complex systems with high-dimensional data or partial observability, specifying the DAG is often infeasible or undesirable due to uncertainty, noise, or computational cost.
To address these challenges, researchers have proposed various abstractions of DAGs that simplify or summarize the underlying causal structure~\citep{Perkovic_2016,Beckers_2019,Anand_2023,Assaad_2023,Assaad_2024,Ferreira_2025}. These abstractions are particularly valuable in fields like epidemiology, where large-scale longitudinal studies often involve repeated measurements over time and many partially observed variables. One such abstraction is the cluster graph~\citep{Anand_2023,Assaad_2023,Assaad_2024,Ferreira_2025}, which groups variables into clusters. A notable and particularly useful instance of this is the summary causal graph (SCG)~\citep{Assaad_2023,Assaad_2024,Ferreira_2025}, where each node represents a time series. SCGs are especially well-suited for applications involving time series or cohort data, where variables are naturally organized over time. These graphs preserve causal relations between time series while abstracting away the time.
While adjustment criteria for causal inference have been well-studied in the context of DAGs\textemdash notably the back-door criterion~\citep{Pearl_1993BackDoor} and its extensions to the generalized adjustment criterion~\citep{Perkovic_2015,Perkovic_2016}\textemdash the theory of adjustment in SCGs remains comparatively underdeveloped. SCGs pose two additional challenges: they may contain cycles, and a single node can represent an entire cluster of variables, both of which complicate the identification of valid adjustment sets.

In recent years, a significant advancement was made for identifying causal effects in SCGs \cite{Assaad_2023,Assaad_2024,Ferreira_2025}. In particular, \cite{Assaad_2024} provided sound complete~\citep{Yvernes_2025} graphical conditions for identifying by covariate adjustment the causal effect of one timepoint $X_{t-\gamma}$ on another timepoint $Y_t$ \textemdash known as a micro causal effect and denoted as $\probac{y_t}{\interv{x_{t-\gamma}}}$~\citep{Ferreira_2025}
\textemdash in SCGs under the assumption of no hidden confounding. However, in \cite{Assaad_2024} only two valid unpracticable adjustment sets were presented for any identifiable micro causal effect. This stands in contrast to the tools developed for DAGs\textemdash like the back-door criterion\textemdash which gives, whenever a causal effect is identifiable by adjustment, multiple valid adjustment sets, each of which is theoretically correct but may differ significantly in practical properties such as statistical efficiency.
In particular, while all valid sets remove confounding bias, they can vary widely in terms of how they impact the variance of the estimated effect~\cite{Rotnitzky_2020,Witte_2020,Smucler_2021,Runge_2021,Henckel_2022}. Therefore, being able to search over adjustment sets to identify the one that yields the most precise estimate (i.e., minimal variance) is critical in real-world applications.

In this paper, we start by simplifying the conditions given in \cite{Assaad_2024} to identify $\probac{y_t}{\interv{x_{t-\gamma}}}$; then we present the SCG-back-door criterion, an extension of the back-door criterion to SCGs, which  provides, whenever $\probac{y_t}{\interv{x_{t-\gamma}}}$ is identifiable by adjustment in SCGs using the conditions given in \cite{Assaad_2024}, a rich collection of valid adjustment sets. 
Furthermore, we identify the quasi-optimal adjustment set among this collection, defined as the one that minimizes the asymptotic variance of the causal effect estimator under standard assumptions in at least some distribution compatible with the SCG; it
is also the smallest set among all valid adjustment that contains as much as possible of all sets that minimize the asymptotic variance of the causal effect estimator in every distribution compatible with the SCG. We emphasize that this set is not necessary the smallest set.
Our results bridge a critical gap in the theory of causal inference from abstracted graphical models and open new directions for variance-aware causal effect estimation in SCGs.

The remaining of the paper is as follows: Section~\ref{sec:preliminaries} introduces needed preliminaries and terminologies. Section~\ref{sec:scg_back-door} simplifies the results given in \cite{Assaad_2024} and introduces the SCG-back-door criterion. Section~\ref{sec:optimal} shows which of the set given by the SCG-back-door criterion is the quasi-optimal. 
Section~\ref{sec:exp}  empirically validates our theoretical findings. Section \ref{sec:discussion} discusses the implications and limitations of our results.
Finally, Section~\ref{sec:conclusion} concludes the paper. All proofs are deferred to the Appendix.
 
\section{Preliminaries and problem setup}
\label{sec:preliminaries}

In this work, we consider a \textit{Discrete-Time Dynamic Structural Causal Model} (DTDSCM), an extension of structural causal models \citep{Pearl_2000} to dynamic systems.
This choice is motivated by the nature of the data encountered in many real-world applications, such as monitoring systems or epidemiological studies, where variables evolve over time and observations are collected at discrete intervals. A DTDSCM allows us to explicitly model temporal dependencies and causal mechanisms operating across time steps, making it well-suited for analyzing interventions. In what follows, we recall the basic formalism of DTDSCM and present the notations and assumptions that will be used throughout the manuscript.
We use the convention that uppercase letters represent variables, lowercase letters represent their values, and letters in blackboard bold represent sets.

\begin{definition}[Discrete-time dynamic structural causal model (DTDSCM)]
	\label{def:DTDSCM}
A discrete-time dynamic structural causal model is a tuple $\mathcal{M}=(\mathbb{L}, \mathbb{V}, \mathbb{F},  \proba{\mathbb{l}})$,
	where
    \begin{itemize}
        \item $\mathbb{L} = \bigcup \{\mathbb{L}^{v^i_t} \mid i \in [1,d], t \in [t_0,t_{max}]\}$ is a set of exogenous variables, which cannot be observed but affect the rest of the model;
        \item 	$\mathbb{V} = \bigcup \{\mathbb{V}^i \mid i \in [1,d]\}$ such that $\forall i \in [1,d]$, $\mathbb{V}^{i} = \{V^{i}_{t} \mid t \in [t_0,t_{max}]\}$, is a set of endogenous variables, which are observed and every $V^i_t \in \mathbb{V}$ is functionally dependent on some subset of $\mathbb{L}^{v^i_t} \cup \mathbb{V}_{\leq t} \backslash \{V^i_t\}$ where $\mathbb{V}_{\leq t} = \{V^j_{t'} \mid j \in [1,d], t'\leq t\}$;
        \item 	$\mathbb{F}$ is a set of functions such that for all $V^i_t \in \mathbb{V}$, $f^{v^i_t}$ is a mapping from $\mathbb{L}^{v^i_t}$ and a subset of $\mathbb{V}_{\leq t} \backslash \{V^i_t\}$ to $V^i_t$;
        \item 	$\proba{\mathbb{l}}$ is a joint probability distribution over $\mathbb{L}$.
    \end{itemize}
\end{definition}

In this definition of DTDSCMs, the set of functions $\mathbb{F}$ encapsulates the causal mechanisms governing the relationships among variables, while $\mathbb{L}$ represents the noise, i.e., the unobserved or hidden variables that affect the observed variables, and  $\mathbb{V}$ represents the observed variables, which we will often refer to as \emph{micro variables}. The uncertainty associated with these unobserved variables is characterized by the distribution $\proba{\mathbb{l}}$, which accounts for the unknown influences outside the observed system. When combined with the causal mechanisms encoded in $\mathbb{F}$, this distribution gives rise to the distribution $\proba{\mathbb{v}}$ over the observed variables $\mathbb{V}$.
Throughout this paper, we make the following two assumptions about the DTDCM.
\begin{assumption}[No hidden confounding]
\label{ass:causal sufficiency}
We assume that 
the exogenous variables are mutually independent, and each variable in $\mathbb{L}$ is at most in one function in $\mathbb{F}$.
\end{assumption}

\begin{assumption}[Causal stationarity]
\label{ass:ctt}
 We assume that, for any causal relation between variables, the associated causal effect remains the same at every time step; in other words, causal dependencies are time-invariant.
\end{assumption}
The first assumption is essential for the theoretical results we derive, while the second is fully required only when working with a single multivariate time series where each time point corresponds to a single observation.
This scenario commonly arises in industrial systems, such as chemical plants, manufacturing lines, or energy production facilities, where various sensors continuously monitor quantities like temperature, pressure, flow rate, vibration, and power consumption. In such tightly controlled and stable environments, the system is expected to operate under stationary conditions, making the stationarity assumption reasonable. This scenario can also arise in health, for example, in Intensive Care Units (ICUs), where we focus on the data of a single patient that are continuously monitored through various sensors measuring variables such as heart rate, blood pressure, oxygen saturation, respiratory rate, and temperature. Over relatively short and clinically stable periods, patient physiology can often be assumed to follow approximately stationary dynamics, making the stationarity assumption acceptable.
Nevertheless, this second assumption could be relaxed to a "weaker"\footnote{By weaker stationarity we mean to say that when dealing with multiple multivariate time series, we only need to assume that all causal relations from $t-\gamma-\gamma_{max}$ to $t$ are time invariant in the FT-DAG but not necessarily in the DTDSCM.} causal stationarity in settings involving multiple multivariate time series\textemdash such as in cohort studies\textemdash where each time point is associated with multiple independent observations. For example, a case where the data of multiple patients in the ICU is collected.
Finally, as most of the literature in causal inference, we  assume that each DTDSCM induces a full-time directed \emph{acyclic} graph (FT-DAG)\footnote{FT-DAGs  are a subclass of DAGs, and similarly, DTDSCMs are instances of standard SCMs. As such, general results and tools developed for DAGs and SCMs are directly applicable to FT-DAGs and DTDSCMs. Throughout this paper, we may occasionally state that a method or result applies to FT-DAGs, even if it was originally established in the broader context of DAGs.}, where every variable in $\mathbb{V}$ corresponds to a node in the graph. In this DAG, a directed edge $\rightarrow$ is drawn from one variable to another if the former serves as an input to the function that determines the latter.  
Formally, FT-DAGs are defined as follows:

\begin{definition}[Full Time Directed Acyclic Graph (FT-DAG)] 
Consider a DTDSCM $\mathcal{M}$. The \emph{full-time directed acyclic graph}{(FT-DAG)} $\mathcal{G} = (\mathbb{V}, \mathbb{E})$ induced by $\mathcal{M}$ is defined in the following way:
	\begin{equation*}
		\begin{aligned}
			\mathbb{E} :=& \{X_{t-\gamma}\rightarrow Y_{t} \mid \forall Y_{t} \in \mathbb{V},~X_{t-\gamma} \in \mathbb{X} \subseteq \mathbb{V}_{\leq t} \backslash \{Y_t\} \st Y_t := f^{y_t}(\mathbb{X}, \mathbb{L}^{y_t}) \text{ in } \mathcal{M}\},\\
		\end{aligned}
	\end{equation*}
\end{definition}


Each node in an FT-DAG represents a timepoint in a time series. The lag $\gamma$ refers to the time difference between the variable observed at times $t$ and $t-\gamma$.
We denote the maximum lag $\gamma_{max}$ as the greatest time lag included in the model.

\textbf {FT-DAG concepts:} Let $\FTCG=(\mathbb{V},\mathbb{E})$ be a graph which is a set of nodes $\mathbb{V}$, and a set of edges $\mathbb{E}$ that  connect different pairs of nodes. Two nodes are adjacent if they are connected by an edge. Nodes denote variables and edges denote relationships between two variables. If the edges point from one node to another, the graph is called a directed graph. A causal graph is a directed graph where edges represent direct causal effects from one variable to another.
A path is a sequence of unique nodes in which each pair of consecutive nodes is adjacent. A directed path from $X$ to $Y$ is a path in which all edges point in the same direction, from the tail node to the head node of the form $ X\rightarrow{\cdots}\rightarrow{Y}$.
A path is said to be blocked by a set of nodes $\mathbb{Z}$ if the set contains a non-collider of the form $\leftarrow V\rightarrow$ or $\rightarrow V \rightarrow$ such that $V\in\mathbb{Z}$, or a collider of the form $\rightarrow V \leftarrow$ such that neither $V$ nor any of its descendants are in $\mathbb{Z}$. 
A back-door path from $X$ to $Y$ is any path between $X$ and $Y$ that starts with an edge pointing into $X$ that is $X\leftarrow$. 
We use the terminology of kinship to denote the relationships between the variables: parents, children, descendants and ancestors. $X$ is a parent of $Y$ and $Y$ is a child of $X$ if $X\rightarrow{Y}$ in $\FTCG$. 
If there is a directed path between $X$ and $Y$, then $X$ is an ancestor of $Y$ and $Y$ is a descendant of $X$. A node is considered as both its own ancestor and its own descendant. We note $\parents{X}{\FTCG}$ parents of $X$ in $\FTCG$, $\children{X}{\FTCG}$ children of $X$ in $\FTCG$, $\descendants{X}{\FTCG}$ descendants of $X$ in $\FTCG$ and $\ancestors{X}{\FTCG}$ ancestors of $X$ in $\FTCG$.

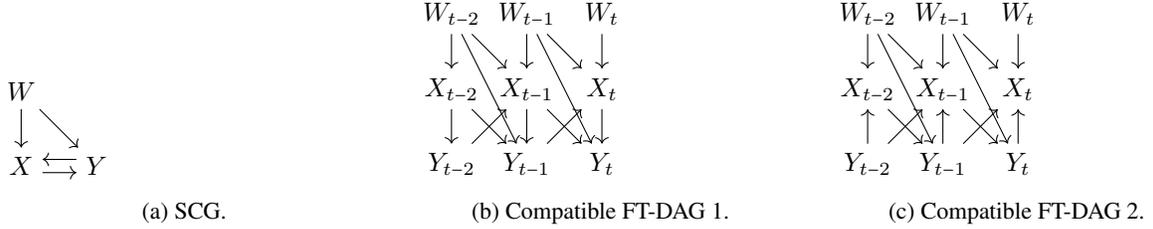
\begin{figure}[t]
    \label{fig:exemple_stress}
    \centering
\hfill
\begin{subfigure}[t]{.3\linewidth}
\begin{tikzpicture}    
    \node (W) {$W$};
    \node (X) [below of=W] {$X$};
    \node (Y) [right of=X] {$Y$};

    \draw [->] (W) -- (X);
    \draw [->] (W) -- (Y);

 \begin{scope}[transform canvas={yshift=-.25em}]
 \draw [->] (X) -- (Y);
 \end{scope}
 \begin{scope}[transform canvas={yshift=.25em}]
 \draw [<-] (X) -- (Y);
 \end{scope}

\end{tikzpicture}
\caption{SCG.}
\end{subfigure}
\hfill
\begin{subfigure}[t]{.3\linewidth}
\begin{tikzpicture}

  \node (X2) {$X_{t-2}$};
  \node (X1) [right of= X2] {$X_{t-1}$};
  \node (X0) [right of= X1] {$X_{t}$};

  \node (T2) [above of= X2] {$W_{t-2}$};
  \node (T1) [right of= T2] {$W_{t-1}$};
  \node (T0) [right of= T1] {$W_{t}$};

  \node (Y2) [below of= X2] {$Y_{t-2}$};
  \node (Y1) [right of= Y2] {$Y_{t-1}$};
  \node (Y0) [right of= Y1] {$Y_{t}$};

    \draw [->] (T2) -- (X2);
    \draw [->] (T2) -- (X1);
    \draw [->] (T1) -- (X1);
    \draw [->] (T1) -- (X0);
    \draw [->] (T0) -- (X0);
    
    \draw [->] (T2) -- (Y1);
    \draw [->] (T1) -- (Y0);

    \draw [->] (X2) -- (Y2);
    \draw [->] (X2) -- (Y1);
    \draw [->] (Y2) -- (X1);
    \draw [->] (X1) -- (Y1);
    \draw [->] (X1) -- (Y0);
    \draw [->] (Y1) -- (X0);
    \draw [->] (X0) -- (Y0);

\end{tikzpicture}
\caption{Compatible FT-DAG 1.}
\end{subfigure}
\hfill
\begin{subfigure}[t]{.3\linewidth}
\begin{tikzpicture}

  \node (X2) {$X_{t-2}$};
  \node (X1) [right of= X2] {$X_{t-1}$};
  \node (X0) [right of= X1] {$X_{t}$};

  \node (T2) [above of= X2] {$W_{t-2}$};
  \node (T1) [right of= T2] {$W_{t-1}$};
  \node (T0) [right of= T1] {$W_{t}$};

  \node (Y2) [below of= X2] {$Y_{t-2}$};
  \node (Y1) [right of= Y2] {$Y_{t-1}$};
  \node (Y0) [right of= Y1] {$Y_{t}$};

    \draw [->] (T2) -- (X2);
    \draw [->] (T2) -- (X1);
    \draw [->] (T1) -- (X1);
    \draw [->] (T1) -- (X0);
    \draw [->] (T0) -- (X0);
    
    \draw [->] (T2) -- (Y1);
    \draw [->] (T1) -- (Y0);

    \draw [<-] (X2) -- (Y2);
    \draw [->] (X2) -- (Y1);
    \draw [->] (Y2) -- (X1);
    \draw [<-] (X1) -- (Y1);
    \draw [->] (X1) -- (Y0);
    \draw [->] (Y1) -- (X0);
    \draw [<-] (X0) -- (Y0);

\end{tikzpicture}
\caption{Compatible FT-DAG 2.}
\end{subfigure}

\caption{Example of an SCG and two compatible FT-DAGs where: $X=$ Sleep Efficiency, $Y=$ Stress , $W=$ Total sleep time~\citep{Jordan_2023}.}
\label{fig:complex}
\end{figure}

In this paper, we focus on causal effect of a variable $X_{t-\gamma}$ on another $Y_t$, denoted as $\probac{y_t}{\interv{x_{t-\gamma}}}$ where the $\interv{}$ operator represents an intervention. This is usually referred to as the micro causal effect\footnote{The term "micro" causal effect is used in contrast to the "macro" causal effect, which refers to the influence of an entire time series on another. In the case of micro causal effects, the focus is on the causal influence of a single time point (or time-specific intervention) on another specific time point within or across time series.}.
In particular, we focus on the technical problem of identifying micro causal effects from observational data\textemdash that is, expressing these effects using a do-free formula involving only observed variables.


Causal effects are generally identifiable under various sets of assumptions. Among the most widely used are those that include the assumptions of conditional exchangeability and positivity~\citep{Hernan_2023, Austin_2011} (these two assumptions together are known as ignorability).
The positivity assumption (\ie, that the probability of receiving each level of treatment is strictly positive for all covariates) is, in principle, testable from data, whereas the conditional exchangeability assumption is not and must instead be justified using background or domain knowledge.
When these assumptions hold, the micro causal effect of a variable $X_{t-\gamma}$ on an outcome $Y_t$ can be identified using the following adjustment formula:
\begin{equation}
\label{eq:adjustment_formula}
\probac{y_t}{\interv{x_{t-\gamma}}}=\sum_{\mathbb{z}}\probac{y_t}{x_{t-\gamma}, \mathbb{z}}\proba{\mathbb{z}},    
\end{equation}
where $\mathbb{Z}$ is a set of covariates that removes confounding bias when adjusted on. In some specific settings, $\mathbb{Z}$ can be considered as the set of confounders.
This adjustment formula is arguably the most well-known do-free expression for causal effects. Many standard estimators, such as ordinary least squares  regression and propensity score methods~\citep{Vable_2019}, are based on this formula.

\citet{Pearl_1993BackDoor} showed that if the underlying  FT-DAG is known, one can test whether the conditional exchangeability assumption holds using a simple graphical condition called the back-door criterion. In practice, there may be multiple sets of covariates $\mathbb{Z}$ that allow for identification via Equation~\ref{eq:adjustment_formula}, and the back-door criterion provides a way to identify such sets.
A set $\mathbb{Z}$ satisfies the back-door criterion relative to an ordered pair $(X_{t-\gamma},Y_t)$ if $\mathbb{Z}$ blocks all back-door paths from $X_{t-\gamma}$ to $Y_t$ and $\mathbb{Z}$ does not contain any descendant of $X_{t-\gamma}$.

If the true FT-DAG is known and if Assumptions~\ref{ass:causal sufficiency} and \ref{ass:ctt} are satisfied, there exists at least one set $\mathbb{Z}$ that satisfies the back-door criterion~\citep{Perkovic_2016}, meaning that in this setup the micro causal effect of $X_{t-\gamma}$ on $Y_t$ is always identifiable  using Equation~\ref{eq:adjustment_formula}. 
However, in practical applications, using FT-DAGs is often impractical\textemdash or even infeasible\textemdash particularly in domains with a large number of variables. Such graphs tend to become overly complex, making them hard to interpret and construct. This challenge is further exacerbated in settings where only partial knowledge of the underlying causal relationships is available, limiting our ability to fully specify an FT-DAG.
For this reason, researchers often resort to a simpler and more summarized representation known as the summary causal graph (SCG)~\citep{Assaad_2022}.

\begin{definition}[Summary Causal Graph (SCG) with no hidden confounding]
	Consider an FT-DAG $\mathcal{G} = (\mathbb{V}, \mathbb{E})$. Under Assumption~\ref{ass:causal sufficiency}, the \emph{summary causal graph (SCG)} $\mathcal{G}^{s} = (\mathbb{V}^s, \mathbb{E}^{{s}})$ compatible with $\mathcal{G}$ is defined in the following way:
	\begin{equation*}
		\begin{aligned}
			\mathbb{V}^s :=& \{V^i = (V^i_{t_0},\cdots,V^i_{t_{max}}) \mid \forall i \in [1, d]\},\\
           \mathbb{E}^{s} :=& \{X\rightarrow Y \mid \forall X,Y \in \mathbb{V}^s,~\exists t'\leq t\in [t_0,t_{max}] 
        \st X_{t'}\rightarrow Y_{t}\in\mathbb{E}\}.
\end{aligned}
\end{equation*}
\end{definition}

An SCG is an abstraction of an FT-DAG where all timepoints in a time series  are compacted into one node, referred to as \emph{macro variable} or \emph{cluster}. 
An arrow $X\rightarrow Y$ exists in the SCG if and only if there is at least one temporal lag $0\leq \gamma\leq\gamma_{max}$ such as $X_{t-\gamma}\rightarrow Y_t$ exists in the FT-DAG. While multiple FT-DAGs may induce the same SCG, the mapping from an FT-DAG to its corresponding SCG is unique. For example, in Figure~\ref{fig:complex} and Figure~\ref{fig:simple}, we provide an SCG with two of its compatible FT-DAGs.
Unlike an FT-DAG, because of its abstract nature, an SCG can contain cycles.
For example, Figure~\ref{fig:complex} illustrates a real-life application where events were recorded at successive time points. It summarizes the causal relationships between sleep efficiency ($X$), stress levels ($Y$), and total sleep time ($W$), each measured over time, into a compact causal graph. In the corresponding FT-DAG, the structure remains acyclic, as each edge respects the temporal order. However, in the associated SCG, a cycle emerges due to the presence of edges such as $X_{t-\gamma} \rightarrow Y_t$ and $Y_{t-\gamma} \rightarrow X_t$. This reflects a plausible real-world dynamic: sleep efficiency on one day influences stress the following day, and conversely, elevated stress levels can impair sleep efficiency the next night. 
Another example is given in the SCG of Figure~\ref{fig:simple} which summarizes the causal relationships between Air pollution levels ($X$), Respiratory hospitals ($Y$), and Traffic density ($W$). In the example, traffic density and air pollution levels naturally exhibit strong temporal persistence. For example, heavy traffic on a given working day is often followed by similarly congested conditions the next day, and air pollution levels often evolve gradually rather than abruptly. These temporal dynamics are captured in the SCG through self-loops on the macro variables representing traffic density ($X$) and air pollution levels ($W$). Such loops encode the idea that the current state of these variables is partly, influenced by their own past values.
In the underlying FT-DAGs, this persistence is encoded by edges of the form $X_{t-\gamma} \rightarrow X_t$ and $W_{t-\gamma} \rightarrow W_t$.

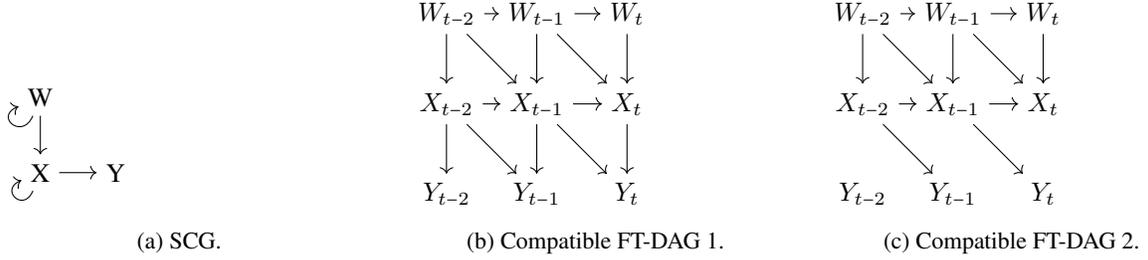
\begin{figure}[t]
    \label{fig:exemple_stress}
    \centering
\hfill
\begin{subfigure}[t]{.3\linewidth}
\begin{tikzpicture}
    
    \node (W) {W};
    \node (X) [below of=W] {X};
    \node (Y) [right of=X] {Y};
    
    \draw [->] (X) -- (Y);
    \draw[->, looseness=4, out=-90-45/2, in=180+45/2] (W) to (W);
    \draw[->, looseness=4, out=-90-45/2, in=180+45/2] (X) to (X);
    \draw [->] (W) -- (X);

\end{tikzpicture}
\caption{SCG.}
\label{fig:simple:SCG}
\end{subfigure}
\hfill
\begin{subfigure}[t]{.3\linewidth}
\begin{tikzpicture}
  \begin{scope}[node distance = 1.2cm]
  \node (A2) {$W_{t-2}$};
  \node (A1) [right of= A2] {$W_{t-1}$};
  \node (A0) [right of= A1] {$W_{t}$};

  \node (X2) [below of= A2] {$X_{t-2}$};
  \node (X1) [right of= X2] {$X_{t-1}$};
  \node (X0) [right of= X1] {$X_{t}$};

  \node (Y2) [below of= X2] {$Y_{t-2}$};
  \node (Y1) [right of= Y2] {$Y_{t-1}$};
  \node (Y0) [right of= Y1] {$Y_{t}$};
  \end{scope}
  
    \draw [->] (A2) -- (X2);
    \draw [->] (A2) -- (X1);
    \draw [->] (A1) -- (X1);
    \draw [->] (A1) -- (X0);
    \draw [->] (A0) -- (X0);

    \draw [->] (X2) -- (Y2);
    \draw [->] (X2) -- (Y1);
    \draw [->] (X1) -- (Y1);
    \draw [->] (X1) -- (Y0);    
    \draw [->] (X0) -- (Y0);

    \draw [->] (X2) -- (X1);
    \draw [->] (X1) -- (X0);
    \draw [->] (A2) -- (A1);
    \draw [->] (A1) -- (A0);

\end{tikzpicture}
\caption{Compatible FT-DAG 1.}
\label{fig:simple:DAG1}
\end{subfigure}
\hfill
\begin{subfigure}[t]{.3\linewidth}
\begin{tikzpicture}
  \begin{scope}[node distance = 1.2cm]
  \node (A2) {$W_{t-2}$};
  \node (A1) [right of= A2] {$W_{t-1}$};
  \node (A0) [right of= A1] {$W_{t}$};

  \node (X2) [below of= A2] {$X_{t-2}$};
  \node (X1) [right of= X2] {$X_{t-1}$};
  \node (X0) [right of= X1] {$X_{t}$};

  \node (Y2) [below of= X2] {$Y_{t-2}$};
  \node (Y1) [right of= Y2] {$Y_{t-1}$};
  \node (Y0) [right of= Y1] {$Y_{t}$};
  \end{scope}
    
    \draw [->] (A2) -- (X2);
    \draw [->] (A2) -- (X1);
    \draw [->] (A1) -- (X1);
    \draw [->] (A1) -- (X0);
    \draw [->] (A0) -- (X0);

    \draw [->] (X2) -- (Y1);
    \draw [->] (X1) -- (Y0);    

    \draw [->] (X2) -- (X1);
    \draw [->] (X1) -- (X0);
  \draw [->] (A2) -- (A1);
    \draw [->] (A1) -- (A0);

\end{tikzpicture}
\caption{Compatible FT-DAG 2.}
\label{fig:simple:DAG2}
\end{subfigure}
\caption{Example of an SCG and two compatible FT-DAGs where: $X=$  Air pollution levels, $Y=$ Respiratory hospital admissions, $W=$ Traffic density~\citep{Zhang_2013, Dominici_2006}.}
\label{fig:simple}
\end{figure}

\textbf {SCG concepts.} Let $\SCG=(\mathbb{V}^s,\mathbb{E}^s)$ be an SCG with a set of macro nodes $\mathbb{V}^s$. 
All concepts introduced for FT-DAGs apply also for SCG. However, for SCGs we also need to introduce additional concepts related to cycles.
A cycle is a directed path that starts and ends at the same node, with all intermediate nodes distinct. A self-loop is an edge that connects a node to itself and is considered a trivial cycle of length one.
A strongly connected component (SCC) in a graph is a maximal set of nodes such that there is a directed path from each node to every other node in the component. In particular, each node always belongs to its own SCC, even if no other node is reachable from or to it. We denote by $\scc{X}{\graph}$ the strongly connected component of node $X$ in the graph $\graph$.
Finally,  to map macro-variables $\mathbb{V}\subseteq \vSCG$ to specific compatible temporal instants we use the following operator: $\temp{\mathbb{W}}{t_2}{t_1} = \tempexplain{\mathbb{W}}{t_1}{t_2}$ where $t_1\leq t_2$.

The back-door criterion cannot be directly applied to identify micro causal effects using the SCG. For example, consider the simple SCG shown in Figure~\ref{fig:simple:SCG}, where $X\rightarrow Y$ and $W\rightarrow X$. If we apply the back-door criterion to this SCG, one might mistakenly conclude that no adjustment is necessary, implying $\probac{y_t}{\interv{x_{t-\gamma}}}=\probac{y_t}{x_{t-\gamma}}$. However, this is incorrect. To understand why, examine the FT-DAG in Figure~\ref{fig:simple:DAG1}, where multiple active back-door paths exist between $X_{t-\gamma}$ and $Y_t$. These paths can be blocked by the set $\{X_{t_0}, \cdots, X_{t-\gamma-1},W_{t_0}, \cdots, W_{t-\gamma}\}$which simultaneously does not contain any descendants of $X_{t-\gamma}$, thus satisfying the back-door criterion. However, using this set in the adjustment formula is impractical in the context of time series, as it would violate the positivity assumption (since the time series correspond to a single observation, making proper adjustment infeasible).
To address this, we focus on cases where we avoid adjusting for variables temporally prior to $t-\gamma-\gamma_{max}$. In this specific example, the set $\{X_{t - \gamma -\gamma_{max}}, \cdots, X_{t-\gamma-1},W_{t - \gamma - \gamma_{max}}, \cdots, W_{t-\gamma}\}$ satisfies the back-door criterion and does not violate the positivity assumption under Assumption~\ref{ass:ctt} (since past values within this range can be reused across time as repeated observations), meaning it can be used to estimate the micro causal effect. However, the back-door criterion alone does not provide a way to identify this set using only the SCG.
This is merely a simple illustration; in practice, scenarios can be far more complex\textemdash such as the one shown in Figure~\ref{fig:complex}\textemdash where identifying the micro causal effect becomes highly non-trivial.

Recently, \cite{Assaad_2024} introduced a set of conditions for identifying micro causal effects using SCGs. These conditions were later proven to be complete for identifiability under Assumptions~\ref{ass:causal sufficiency} and \ref{ass:ctt} by \cite{Yvernes_2025}. Initially, the formulation in \cite{Assaad_2024} was relatively complex, requiring the enumeration of all active paths between $X_{t-\gamma}$ and $Y_t$. An equivalent but somewhat simpler formulation was later provided in \cite{Yvernes_2025}, although it still involved enumerating certain types of paths and introduced technically involved definitions.
In addition to these conditions, \cite{Assaad_2024} proposed two valid sets of covariates that can be used in the adjustment formula whenever the identifiability conditions are satisfied. These sets are:

$$\mathbb{A}^1=\temp{De(X,\mathcal{G}^s)}{t-\gamma-1}{t-\gamma-1-\gamma_{max}} \cup \temp{\mathbb{V}\backslash De(X,\mathcal{G}^s)}{t-\gamma}{t-\gamma-\gamma_{max}}$$
and
$$\mathbb{A}^2=\temp{(An(\{X,Y\},\mathcal{G}^s))\cap De(X,\mathcal{G}^s)}{t-\gamma-1}{t-\gamma-1-\gamma_{max}} \cup \temp{An(\{X,Y\},\mathcal{G}^s)\backslash De(X,\mathcal{G}^s)}{t-\gamma}{t-\gamma-\gamma_{max}}.$$

In many applications, it would be beneficial to have more valid adjustment sets, allowing the user to choose the most appropriate one based on factors such as the validity of the positivity assumption or other constraints, like missing data. Furthermore, having multiple adjustment sets paves the way for finding the optimal set, which minimizes the asymptotic variance of the micro causal effect estimator, as characterized in \cite{Henckel_2022} for the case when the true FT-DAG is known. The results presented in \cite{Henckel_2022} were originally established in the context of linear models using the ordinary least squares estimator. However, subsequent work by \cite{Rotnitzky_2020} demonstrated that these results also hold in more general non-parametric settings and a class of non-parametric estimators.
We emphasize that the characterization of the optimal set given in \cite{Henckel_2022} only applies to FT-DAG and not to SCGs. In fact, searching for the optimal adjustment set turned out to be hard as it will be explained in Section~\ref{sec:optimal}. A simpler endeavor would be to search for the quasi-optimal set which will be defined formally in Section~\ref{sec:optimal}.

Now that we have introduced the necessary terminology and preliminaries, we state the three problems that we aim to solve.

\begin{problem}
We aim to present a formulation equivalent to the conditions given in \cite{Assaad_2024} and in \cite{Yvernes_2025}, but in a simpler (does not require enumerating paths) and more interpretable (only based on kinship relations) form, for identifying micro causal effects from SCGs.
\end{problem}

\begin{problem}
We aim to extend the back-door criterion to SCGs, providing a criterion that, under the conditions outlined in \cite{Assaad_2024},  offers multiple sets of covariates that can be used in the adjustment formula.
\end{problem}

\begin{problem}
We aim to characterize the quasi-optimal set, among those found by our extension of the back-door criterion.
\end{problem}

\section{A criterion for finding adjustment sets for identifying causal effects in SCGs}
\label{sec:scg_back-door}

We start this section by presenting a simplified version of the identifiability theorem for micro causal effects in SCGs, originally introduced in \cite{Assaad_2024} and later proven complete in \cite{Yvernes_2025}.

\begin{restatable}{theorem}{theoremIdentifiabilityTotalEffect}
\label{theorem:identifiability_total_effect_SCG}
Consider an SCG ${\SCG}=({\vSCG},{\eSCG})$. The micro causal effect $\probac{y_t}{\interv{x_{t-\gamma}}}$ is identifiable if and only if $X \notin \ancestors{Y}{\SCG}$ or $X\in \ancestors{Y}{\SCG}$ and one of the following holds
\begin{enumerate}[label=\textbf{\Alph*}]
    \item \label{cond1} $\scc{X}{\SCG} = \{X\}$, or
    \item \label{cond2} $\gamma=0$ and $\scc{X}{\SCG} \neq \{X\}$ and $\ancestors{Y}{\SCG\backslash\{X\}} \cap \scc{X}{\SCG} = \emptyset$, or
    \item \label{cond3} $\gamma = 1$ and $\cycles{Y}{\SCG} =\{X\rightleftarrows Y\}$\footnote{Condition~\ref{cond3} can be replaced by $\gamma = 1$ and $\scc{X}{\SCG} \subseteq \{X,Y\}$ and $\selfloop{Y} \notin \eSCG$}.
\end{enumerate}
\end{restatable}

This version of the theorem offers a simplification over the formulations presented in \cite{Assaad_2024} and \cite{Yvernes_2025}. Unlike those, it avoids the need to enumerate activated paths in the SCG, relying instead on kinship relations, which makes it both more interpretable and conceptually clearer. 
For instance, it is straightforward to verify that the micro causal effect $\probac{y_t}{\interv{x_{t-\gamma}}}$ is identifiable in the SCG shown in Figure~\ref{fig:identifiableA}, as it satisfies Condition~\ref{cond1}. Similarly, identifiability also holds in the SCG depicted in Figure~\ref{fig:identifiableB} due to the fulfillment of Condition~\ref{cond2}.
Finally, the SCG in Figure~\ref{fig:complex} satisfies the Condition~\ref{cond3} of the theorem, implying that the micro causal effect, is identifiable in that setting as well.

Additionally, as we will show when we introduce the SCG-back-door criterion, the structure of this version of the theorem is particularly well-suited for systematically identifying multiple valid adjustment sets.
Before introducing the SCG-back-door criterion, which  enumerate valid adjustment sets whenever $\probac{y_t}{\interv{x_{t-\gamma}}}$ is identifiable according to Theorem~\ref{theorem:identifiability_total_effect_SCG}, we introduce additional notions that will be useful to define the criterion. We consider the set of possible descendants of a temporal variable $V_{t_1}$ in an SCG as the set of temporal variables that are descendants of $V_{t_1}$ in at least one compatible FT-DAG and we denoted as $\posdescendants{V_{t_1}}{\SCG}$. Note that the difference between $\temp{\descendants{V}{\SCG}}{t_{2}}{t_1}$ and $\posdescendants{V_{t_1}}{\SCG}$ lies in the fact that $V_{t_1+1}$ is always in $\temp{\descendants{V}{\SCG}}{t_{2}}{t_1}$ whereas $V_{t_1+1}$ is in $\posdescendants{V_{t_1}}{\SCG}$ if and only if there exists a cycle on $V$ in $\SCG$.
As in \cite{Henckel_2022} we define the causal nodes $\cn{X^*}{Y^*}{\graph}$ of a any pair of nodes $(X^*, Y^*)$\textemdash be it a pair of macro nodes or micro nodes\textemdash in a any graph $\graph$\textemdash be it an FT-DAG or an SCG\textemdash as the nodes on a directed path from $X^*$ to $Y^*$, excluding $X^*$.
Moreover, for any graph $\graph$ and any pair of nodes $(X^*, Y^*)$, we define the extended causal nodes $\ecn{X^*}{Y^*}{\graph}$ as the causal nodes and their strongly connected components.
$$\cn{X^*}{Y^*}{\graph} = \{V^* \in \pi \mid \forall \text{ path }\pi = \langle X^*\rightarrow\cdots\rightarrow Y^*\rangle \}\backslash\{X^*\}$$
$$\ecn{X^*}{Y^*}{\graph} = \bigcup\limits_{V^*\in\cn{X^*}{Y^*}{\graph}}\scc{V^*}{\graph}.$$

\begin{figure}[t]
    \label{fig:exemple_stress}
    \centering
\hfill
\begin{subfigure}[t]{.23\linewidth}
\begin{tikzpicture}
SCG
[
    ->,                       
    >=Stealth,                
    node distance=3cm,        
    every node/.style={circle, draw, minimum size=1cm}]
    
    \node (W) {W};
    \node (X) [below of=W] {X};
    \node (Y) [right of=X] {Y};
    \node (U) [above of=Y] {U};

    \draw [->] (X) -- (Y);
    \draw[->, looseness=4, out=-90-45/2, in=180+45/2] (W) to (W);
    \draw[->, looseness=4, out=-90-45/2, in=180+45/2] (X) to (X);
        \draw[->, looseness=4, out=-90+45/2, in=-45/2] (Y) to (Y);
    \draw[->, looseness=4, out=-90+45/2, in=-45/2] (U) to (U);

    \draw [->] (W) -- (X);
    \draw [->] (W) -- (Y);

    \draw [->] (X) -- (U);
    \draw [->] (U) -- (Y);

\end{tikzpicture}
\caption{SCG 1.}
\end{subfigure}
\hfill
\begin{subfigure}[t]{.23\linewidth}
\begin{tikzpicture}
SCG
[
    ->,                       
    >=Stealth,                
    node distance=3cm,        
    every node/.style={circle, draw, minimum size=1cm}]
    
    \node (W) {W};
    \node (X) [below of=W] {X};
    \node (Y) [right of=X] {Y};
    \node (U) [above of=Y] {U};
    
    \draw [->] (X) -- (Y);
    \draw[->, looseness=4, out=-90-45/2, in=180+45/2] (W) to (W);
    \draw[->, looseness=4, out=-90-45/2, in=180+45/2] (X) to (X);
    \draw[->, looseness=4, out=-90+45/2, in=-45/2] (Y) to (Y);
    \draw[->, looseness=4, out=-90+45/2, in=-45/2] (U) to (U);
    \draw [->] (W) -- (X);
    \draw [->] (U) -- (Y);

 \begin{scope}[transform canvas={yshift=-.25em}]
 \draw [->] (W) -- (U);
 \end{scope}
 \begin{scope}[transform canvas={yshift=.25em}]
 \draw [<-] (W) -- (U);
 \end{scope}

\end{tikzpicture}
\caption{SCG 2.}
\end{subfigure}
\hfill
\begin{subfigure}[t]{.23\linewidth}
\begin{tikzpicture}
SCG
[
    ->,                       
    >=Stealth,                
    node distance=3cm,        
    every node/.style={circle, draw, minimum size=1cm}]
    
    \node (W) {W};
    \node (X) [below of=W] {X};
    \node (Y) [right of=X] {Y};
    \node (U) [above of=Y] {U};
    
    \draw [->] (X) -- (Y);
    \draw[->, looseness=4, out=-90-45/2, in=180+45/2] (W) to (W);
    \draw[->, looseness=4, out=-90-45/2, in=180+45/2] (X) to (X);
    \draw[->, looseness=4, out=-90+45/2, in=-45/2] (Y) to (Y);
    \draw[->, looseness=4, out=-90+45/2, in=-45/2] (U) to (U);
    \draw [->] (X) -- (W);

 \begin{scope}[transform canvas={xshift=-.25em}]
 \draw [->] (Y) -- (U);
 \end{scope}
 \begin{scope}[transform canvas={xshift=.25em}]
 \draw [<-] (Y) -- (U);
 \end{scope}

 \begin{scope}[transform canvas={yshift=-.25em}]
 \draw [->] (W) -- (U);
 \end{scope}
 \begin{scope}[transform canvas={yshift=.25em}]
 \draw [<-] (W) -- (U);
 \end{scope}

\end{tikzpicture}\caption{SCG 3.}
\end{subfigure}
\caption{Examples of SCGs where $\probac{y_t}{\interv{x_{t-\gamma}}}$ is identifiable because Condition~\ref{cond1} of Theorem~\ref{theorem:identifiability_total_effect_SCG} is satisfied.}
\label{fig:identifiableA}
\end{figure}
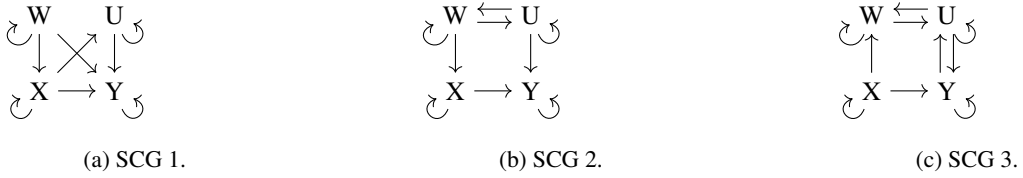

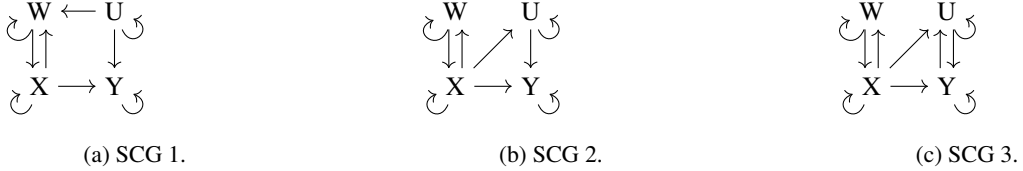
\begin{figure}[t]
    \label{fig:exemple_stress}
    \centering
\hfill
\begin{subfigure}[t]{.23\linewidth}
\begin{tikzpicture}
SCG
[
    ->,                       
    >=Stealth,                
    node distance=3cm,        
    every node/.style={circle, draw, minimum size=1cm}]
    
    \node (W) {W};
    \node (X) [below of=W] {X};
    \node (Y) [right of=X] {Y};
    \node (U) [above of=Y] {U};

    \draw [->] (X) -- (Y);
    \draw[->, looseness=4, out=-90-45/2, in=180+45/2] (W) to (W);
    \draw[->, looseness=4, out=-90-45/2, in=180+45/2] (X) to (X);
        \draw[->, looseness=4, out=-90+45/2, in=-45/2] (Y) to (Y);
    \draw[->, looseness=4, out=-90+45/2, in=-45/2] (U) to (U);

 \begin{scope}[transform canvas={xshift=-.25em}]
 \draw [->] (W) -- (X);
 \end{scope}
 \begin{scope}[transform canvas={xshift=.25em}]
 \draw [<-] (W) -- (X);
 \end{scope}

    \draw [->] (U) -- (W);
    \draw [->] (U) -- (Y);

\end{tikzpicture}
\caption{SCG 1.}
\end{subfigure}
\hfill
\begin{subfigure}[t]{.23\linewidth}
\begin{tikzpicture}
SCG
[
    ->,                       
    >=Stealth,                
    node distance=3cm,        
    every node/.style={circle, draw, minimum size=1cm}]
    
    \node (W) {W};
    \node (X) [below of=W] {X};
    \node (Y) [right of=X] {Y};
    \node (U) [above of=Y] {U};

    \draw [->] (X) -- (Y);
    \draw[->, looseness=4, out=-90-45/2, in=180+45/2] (W) to (W);
    \draw[->, looseness=4, out=-90-45/2, in=180+45/2] (X) to (X);
        \draw[->, looseness=4, out=-90+45/2, in=-45/2] (Y) to (Y);
    \draw[->, looseness=4, out=-90+45/2, in=-45/2] (U) to (U);

 \begin{scope}[transform canvas={xshift=-.25em}]
 \draw [->] (W) -- (X);
 \end{scope}
 \begin{scope}[transform canvas={xshift=.25em}]
 \draw [<-] (W) -- (X);
 \end{scope}

    \draw [->] (X) -- (U);
    \draw [->] (U) -- (Y);

\end{tikzpicture}
\caption{SCG 2.}
\end{subfigure}
\hfill
\begin{subfigure}[t]{.23\linewidth}
\begin{tikzpicture}
SCG
[
    ->,                       
    >=Stealth,                
    node distance=3cm,        
    every node/.style={circle, draw, minimum size=1cm}]
    
    \node (W) {W};
    \node (X) [below of=W] {X};
    \node (Y) [right of=X] {Y};
    \node (U) [above of=Y] {U};

    \draw [->] (X) -- (Y);
    \draw[->, looseness=4, out=-90-45/2, in=180+45/2] (W) to (W);
    \draw[->, looseness=4, out=-90-45/2, in=180+45/2] (X) to (X);
        \draw[->, looseness=4, out=-90+45/2, in=-45/2] (Y) to (Y);
    \draw[->, looseness=4, out=-90+45/2, in=-45/2] (U) to (U);

 \begin{scope}[transform canvas={xshift=-.25em}]
 \draw [->] (W) -- (X);
 \end{scope}
 \begin{scope}[transform canvas={xshift=.25em}]
 \draw [<-] (W) -- (X);
 \end{scope}

    \draw [->] (X) -- (U);
 \begin{scope}[transform canvas={xshift=-.25em}]
 \draw [->] (Y) -- (U);
 \end{scope}
 \begin{scope}[transform canvas={xshift=.25em}]
 \draw [<-] (Y) -- (U);
 \end{scope}  
\end{tikzpicture}
\caption{SCG 3.}
\end{subfigure}
\caption{Examples of SCGs where $\probac{y_t}{\interv{x_{t}}}$ is identifiable because Condition~\ref{cond2} of Theorem~\ref{theorem:identifiability_total_effect_SCG} is satisfied.}
\label{fig:identifiableB}
\end{figure}

\begin{definition}[SCG-back-door criterion for a micro causal effect]
\label{def:backdoor}
Consider an SCG with no hidden confounding ${\SCG}=({\vSCG},{\eSCG})$. If $X\in \ancestors{Y}{\SCG}$ then a set $\mathbb{Z}$ satisfies the SCG-back-door criterion relative to a pair of variables $(X_{t - \gamma}, Y_t)$ compatible with an ordered pair of nodes $(X, Y)$ in $\SCG$ if $\mathbb{D}\cap\mathbb{Z}=\emptyset$ where $\mathbb{D}=\posdescendants{X_{t-\gamma}}{\SCG}$, $\forall V_{t'} \in \mathbb{Z},~t-\gamma-\gamma_{max}\leq t'\leq t$, and

\begin{itemize}[align=left]
    \item[Condition~\ref{cond1} is satisfied] \label{adjcond1} and one of the following holds:
    \begin{enumerate}
        \item \label{adjcond1a} $\mathbb{P}\backslash\mathbb{D} \subseteq \mathbb{Z}$ where
            \begin{itemize}
                \item $\mathbb{P} = \temp{\parents{\scc{X}{\SCG}}{\SCG}}{t-\gamma}{t-\gamma-\gamma_{max}}$, and
            \end{itemize}
        \item \label{adjcond1b} $\mathbb{P} \backslash \mathbb{D} \subseteq \mathbb{Z}$ and
            $\cycles{X}{\SCG} = \emptyset$
            where
            \begin{itemize}
                \item $\mathbb{P} = \temp{\parents{\ecn{X}{Y}{\SCG}}{\SCG}}{t}{t-\gamma-\gamma_{max}}$, and
            \end{itemize}
        \item \label{adjcond1c} $\mathbb{Z} = \mathbb{Z}_1\sqcup\mathbb{Z}_2$ and
            $\gamma=0$
            where
            \begin{itemize}
                \item $\ecnbd{X}{Y}{\SCG}{\mathbb{Z}_2} =\{V \in \ecn{X}{Y}{\SCG} \mid \exists \text{ a backdoor path } \pi_{BD}=\langle X \leftarrow \cdots Y\rangle \in \SCG \st V \in \pi_{BD} \text{ and } \forall \text{ collider }\langle \cdot \rightarrow C \leftarrow \cdot \rangle \subseteq \pi_{BD}, \exists t_C \st C_{t_C}\in\mathbb{Z}_2\}$, and
                \item $\mathbb{Z}_1 = \left( \temp{\parents{\cn{X}{Y}{\SCG}}{\SCG}}{t}{t-\gamma-\gamma_{max}} \cup \temp{\parents{\ecnbd{X}{Y}{\SCG}{\mathbb{Z}_2}}{\SCG}}{t}{t-\gamma-\gamma_{max}}\right)\backslash \mathbb{D}$, and
            \end{itemize}
        \item \label{adjcond1d} $\mathbb{P} \backslash \mathbb{D} \subseteq \mathbb{Z}$ and $\cycles{X}{\SCG} \neq \emptyset$ and $\gamma>0$ where
        \begin{itemize}
            \item $\mathbb{P} = \temp{\parents{X}{\SCG}}{t}{t-\gamma-\gamma_{max}+1} \cup \temp{\parents{\ecn{X}{Y}{\SCG}}{\SCG}}{t}{t-\gamma-\gamma_{max}}$
        \end{itemize}
    \end{enumerate}
    \item[Condition~\ref{cond2} is satisfied] \label{adjcond2} and one of the following holds:
    \begin{enumerate}
        \item \label{adjcond2a} $\mathbb{P}\backslash\mathbb{D} \subseteq \mathbb{Z}$ where
            \begin{itemize}
                \item $\mathbb{P} = \temp{\parents{\scc{X}{\SCG}}{\SCG}}{t-\gamma}{t-\gamma-\gamma_{max}}$, and
            \end{itemize}
        \item \label{adjcond2b} $\mathbb{Z} = \mathbb{Z}_1\sqcup\mathbb{Z}_2$
            where
            \begin{itemize}
                \item $\ecnbd{X}{Y}{\SCG}{\mathbb{Z}_2} =\{V \in \ecn{X}{Y}{\SCG} \mid \exists \text{ a backdoor path } \pi_{BD}=\langle X \leftarrow \cdots Y\rangle \in \SCG \st V \in \pi_{BD} \text{ and } \forall \text{ collider }\langle \cdot \rightarrow C \leftarrow \cdot \rangle \subseteq \pi_{BD}, \exists t_C \st C_{t_C}\in\mathbb{Z}_2\}$, and
                \item $\mathbb{Z}_1 = \left( \temp{\parents{\cn{X}{Y}{\SCG}}{\SCG}}{t}{t-\gamma-\gamma_{max}} \cup \temp{\parents{\ecnbd{X}{Y}{\SCG}{\mathbb{Z}_2}}{\SCG}}{t}{t-\gamma-\gamma_{max}}\right)\backslash \mathbb{D}$, and
            \end{itemize}
    \end{enumerate}
    \item[Condition~\ref{cond3} is satisfied] \label{adjcond3} and
        $\left( \mathbb{P}_X \cup \mathbb{P}_Y\right) \backslash \mathbb{D} \subseteq \mathbb{Z}$, and
            ($\mathbb{P}_X^{all} \subseteq \mathbb{Z}$, or
            $\mathbb{P}_Y^{all} \subseteq \mathbb{Z}$)
        where
        \begin{itemize}
            \item $\mathbb{P}_X = \temp{\parents{X}{\SCG}}{t}{t-\gamma_{max}}$
            \item $\mathbb{P}_Y = \temp{\parents{Y}{\SCG}}{t}{t-\gamma_{max}}$
            \item $\mathbb{P}_X^{all} = \temp{\parents{X}{\SCG}}{t-\gamma-\gamma_{max}}{t-\gamma-\gamma_{max}}$
            \item $\mathbb{P}_Y^{all} = \temp{\parents{Y}{\SCG}}{t-\gamma-\gamma_{max}}{t-\gamma-\gamma_{max}}$.
        \end{itemize}
\end{itemize}
\end{definition}

The name "SCG-back-door" reflects the fact that any set that satisfies this criterion in a given SCG satisfies the classical back-door criterion in every FT-DAG that is compatible with the SCG. We organized Conditions \ref{cond1} and \ref{cond2} in Definition \ref{def:backdoor} to follow a specific logical order. The intuition behind this structure is as follows. In the first item, we aim to block all back-door paths from the side of the treatment, by adjusting for variables that have an arrow pointing into the strongly connected component of $X$. In the last items, we shift focus to the outcome side, adjusting for variables that are parents of the causal nodes or extended causal nodes.
As for Condition \ref{cond3}, which is more restrictive, it requires including all variables that are parents of either $X$ or $Y$ over a restricted time window. Outside of this window, one can choose to block back-door paths either from the side of $X$ or from the side of $Y$, following the same approach as in Conditions \ref{cond1} and \ref{cond2}.

These conditions allow us to explore a broader range of adjustment sets compared to those presented in \cite{Assaad_2024}, which only propose two possible sets: $\mathbb{A}^1$, the largest acceptable set that includes all temporal variables except those that are descendants of the treatment, and a second $\mathbb{A}^2$, shorter alternative set that considers only the ancestors of the treatment and the outcome that are not descendants of the treatment $X_{t-\gamma}$. 
In particular, both adjustment sets $\mathbb{A}^1$ and $\mathbb{A}^2$ can be recovered using the SCG-back-door criterion.
Regarding $\mathbb{A}^1$, since it includes all temporal variables that are not descendants of the treatment, it necessarily contains the variables allowed by the first item in Definition~\ref{def:backdoor}.
In fact, the set $\mathbb{P}\backslash\mathbb{D}$ defined in the first item is a subset of $\mathbb{A}^1$, and extending it to include all such non-descendant temporal variables yields exactly $\mathbb{A}^1$.
Indeed, $\mathbb{A}^1$ satisfies the overall conditions of the SCG-back-door criterion: $\posdescendants{X_{t-\gamma}}{\SCG} \cap \mathbb{Z} = \emptyset$, and $\forall V_{t'} \in \mathbb{Z},~ t - \gamma - \gamma_{\max} \leq t' \leq t$.
For example, in Figure~\ref{fig:simple}, with $\gamma = 1$ and $\gamma_{\max} = 1$, we obtain $\mathbb{A}^1 = \{W_{t-2}, W_{t-1}, W_{t}, X_{t-2}, Y_{t-2}, Y_{t-1}\}$, which contains all variables satisfying the above conditions.
In contrast, the first item of the SCG-back-door criterion accepts, for example, the set $\mathbb{Z} = \{X_{t-2}, W_{t-2}, W_{t-1}\}$, showing that $\mathbb{Z} \subseteq \mathbb{A}^1$.
A similar argument holds for $\mathbb{A}^2$.
Since $\mathbb{A}^2$ includes all possible temporal ancestors of $X_{t-\gamma}$ and $Y_t$ that are not descendants $X_{t-\gamma}$, it necessarily contains the adjustment sets $\mathbb{P}\backslash\mathbb{D}$ in the first item.
Extending $\mathbb{P}\backslash\mathbb{D}$ from the first item to include all remaining ancestors of $X$ and those of $Y$ reconstructs $\mathbb{A}^2$ precisely.
The same intuition holds for the other items of the SCG-back-door criterion.
Thus,  each item in our SCG-back-door criterion returns a large enough collection of adjustment sets that include $\mathbb{A}^1$ or $\mathbb{A}^2$. In addition, all these adjustment sets are subsets of  $\mathbb{A}^1$.

The following theorem formally establishes the correctness of the SCG-back-door criterion for identifying by adjustment the micro causal effect.

\begin{restatable}{theorem}{theoremSCGbackdoor}
    Given an SCG $\SCG$, an ordered pair of variables $(X,Y)$ and a lag $\gamma$, if $\mathbb{Z}$ satisfies the SCG-back-door criterion and if, for all values $\mathbb{z}$ of variables $\mathbb{Z}$, $\probac{x_{t-\gamma}}{\mathbb{z}}>0$ then the micro causal effect of $X_{t-\gamma}$ on $Y_t$ is identifiable and is given by Equation~\ref{eq:adjustment_formula}. 
\end{restatable}

We note that the condition $\posdescendants{X_{t-\gamma}}{\SCG}\cap\mathbb{Z}=\emptyset$ in the SCG-back-door criterion is too restrictive. This restriction reflects our focus on extending the back-door criterion, which eventhough it is useful at enumerating many adjustment sets, it\textemdash unlike the generalized adjustment criterion\textemdash does not aim to enumerate \emph{all} possible adjustment sets. 
In addition, there are cases where the SCG-back-door criterion fails to capture certain valid adjustment sets. Consider one FT-DAG compatible with the SCG shown in Figure~\ref{fig:incompletness:SCG}, where the goal is to estimate the effect of $X_t$ on $Y_t$. Adjusting for the variables highlighted in red in Figures~\ref{fig:incompletness:DAG1} and \ref{fig:incompletness:DAG2} satisfies items 1 and 2 of Condition~\ref{cond1} of the SCG-back-door criterion, respectively.
However, if the true FT-DAG were known, it would be possible to identify alternative valid adjustment sets that are not selected by the criterion. For example, in Figure~\ref{fig:incompletness:DAG2}, the back-door paths initially blocked by adjusting for $R_t$ could equally be blocked by adjusting for $U_{t-1}$ and $U_t$ as in Figure ~\ref{fig:incompletness:DAG4}. 
However, interestingly, as we will show in the next section, although the SCG-back-door criterion does not identify every single valid adjustment set, it does generate a sufficient collection that includes the quasi-optimal adjustment set.

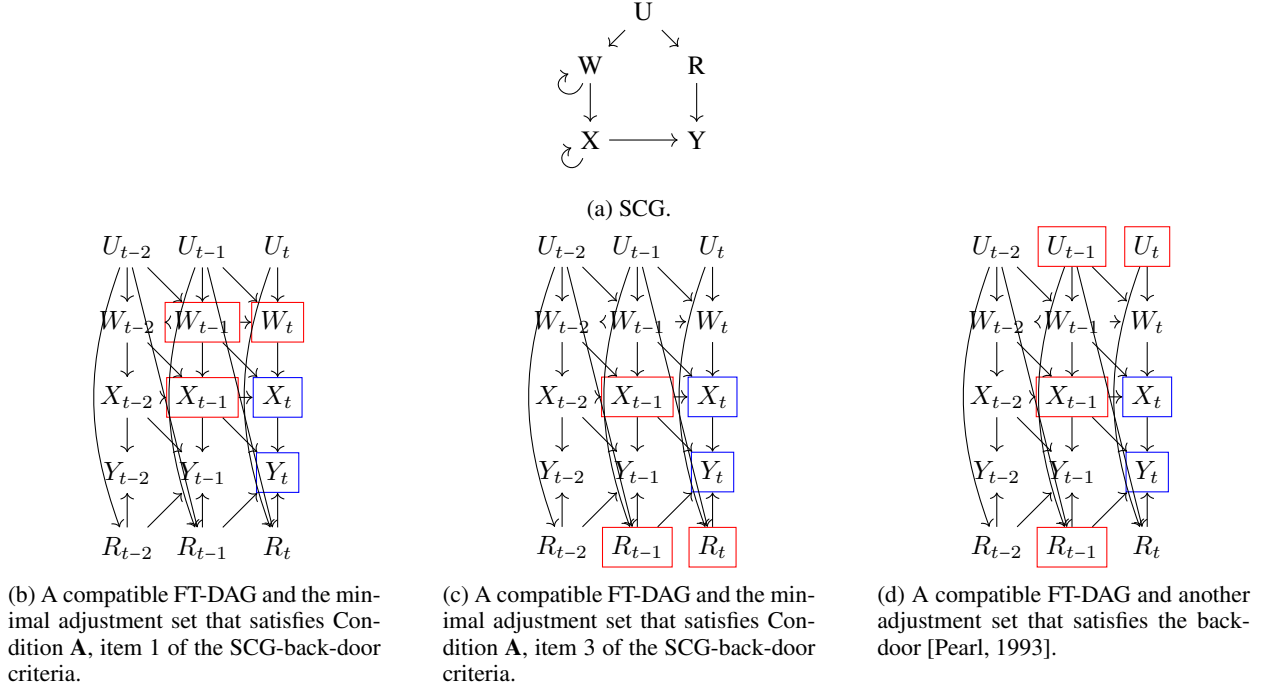
\begin{figure}[t]
    \centering

\begin{subfigure}[t]{\linewidth}
\centering
\begin{tikzpicture}
    
    \node (U) {U};
    \node (W) [below left of= U] {W};
    \node (R) [below right of= U] {R};
    \node (X) [below of= W] {X};
    \node (Y) [below of= R] {Y};
    
    \draw [->] (X) -- (Y);
    \draw[->, looseness=4, out=-90-45/2, in=180+45/2] (X) to (X);
    \draw [->] (W) -- (X);
    \draw [->] (R) -- (Y);
    \draw [->] (U) -- (R);
    \draw [->] (U) -- (W);
    \draw[->, looseness=4, out=-90-45/2, in=180+45/2] (W) to (W);

\end{tikzpicture}
\caption{\centering SCG.}
\label{fig:incompletness:SCG}
\end{subfigure}

\begin{subfigure}[t]{.3\linewidth}
\centering
\begin{tikzpicture}

  \node (U2) {$U_{t-2}$};
  \node (U1) [right of= U2] {$U_{t-1}$};
  \node (U0) [right of= U1] {$U_{t}$};

  \node (A2) [below of= U2]{$W_{t-2}$};
  \node [draw=red] (A1) [right of= A2] {$W_{t-1}$};
  \node [draw=red] (A0) [right of= A1] {$W_{t}$};

  \node (X2) [below of= A2] {$X_{t-2}$};
  \node [draw=red] (X1) [right of= X2] {$X_{t-1}$};
  \node [draw=blue] (X0) [right of= X1] {$X_{t}$};

  \node (Y2) [below of= X2] {$Y_{t-2}$};
  \node (Y1) [right of= Y2] {$Y_{t-1}$};
  \node [draw=blue] (Y0) [right of= Y1] {$Y_{t}$};

  \node (V2) [below of= Y2] {$R_{t-2}$};
  \node (V1) [right of= V2] {$R_{t-1}$};
  \node (V0) [right of= V1] {$R_{t}$};

    \draw [->] (U2) -- (A2);
    \draw [->] (U2) -- (A1);
    \draw [->] (U1) -- (A1);
    \draw [->] (U1) -- (A0);
    \draw [->] (U0) -- (A0);

    \draw[->, bend right=20] (U2) to (V2);
    \draw[->] (U2) to (V1);
    \draw[->, bend right=20] (U1) to (V1);
    \draw[->] (U1) to (V0);
    \draw[->, bend right=20] (U0) to (V0);

    \draw [->] (A2) -- (X2);
    \draw [->] (A2) -- (X1);
    \draw [->] (A1) -- (X1);
    \draw [->] (A1) -- (X0);
    \draw [->] (A0) -- (X0);

    \draw [->] (X2) -- (Y2);
    \draw [->] (X2) -- (Y1);
    \draw [->] (X1) -- (Y1);
    \draw [->] (X1) -- (Y0);    
    \draw [->] (X0) -- (Y0);

    \draw [->] (X2) -- (X1);
    \draw [->] (X1) -- (X0);
    \draw [->] (A2) -- (A1);
    \draw [->] (A1) -- (A0);

    \draw [->] (V2) -- (Y2);
    \draw [->] (V2) -- (Y1);
    \draw [->] (V1) -- (Y1);
    \draw [->] (V1) -- (Y0);
    \draw [->] (V0) -- (Y0);

\end{tikzpicture}
\caption{A compatible FT-DAG and the minimal adjustment set that satisfies Condition~\ref{cond1}, item~\ref{adjcond1a} of the SCG-back-door criteria.}
\label{fig:incompletness:DAG1}
\end{subfigure}
\hfill
\begin{subfigure}[t]{.3\linewidth}
\centering
\begin{tikzpicture}

  \node (U2) {$U_{t-2}$};
  \node (U1) [right of= U2] {$U_{t-1}$};
  \node (U0) [right of= U1] {$U_{t}$};

  \node (A2) [below of= U2]{$W_{t-2}$};
  \node (A1) [right of= A2] {$W_{t-1}$};
  \node (A0) [right of= A1] {$W_{t}$};

  \node (X2) [below of= A2] {$X_{t-2}$};
  \node [draw=red] (X1) [right of= X2] {$X_{t-1}$};
  \node [draw=blue] (X0) [right of= X1] {$X_{t}$};

  \node (Y2) [below of= X2] {$Y_{t-2}$};
  \node (Y1) [right of= Y2] {$Y_{t-1}$};
  \node [draw=blue] (Y0) [right of= Y1] {$Y_{t}$};

  \node (V2) [below of= Y2] {$R_{t-2}$};
  \node [draw=red] (V1) [right of= V2] {$R_{t-1}$};
  \node [draw=red] (V0) [right of= V1] {$R_{t}$};

    \draw [->] (U2) -- (A2);
    \draw [->] (U2) -- (A1);
    \draw [->] (U1) -- (A1);
    \draw [->] (U1) -- (A0);
    \draw [->] (U0) -- (A0);

    \draw[->, bend right=20] (U2) to (V2);
    \draw[->] (U2) to (V1);
    \draw[->, bend right=20] (U1) to (V1);
    \draw[->] (U1) to (V0);
    \draw[->, bend right=20] (U0) to (V0);

    \draw [->] (A2) -- (X2);
    \draw [->] (A2) -- (X1);
    \draw [->] (A1) -- (X1);
    \draw [->] (A1) -- (X0);
    \draw [->] (A0) -- (X0);

    \draw [->] (X2) -- (Y2);
    \draw [->] (X2) -- (Y1);
    \draw [->] (X1) -- (Y1);
    \draw [->] (X1) -- (Y0);    
    \draw [->] (X0) -- (Y0);

    \draw [->] (X2) -- (X1);
    \draw [->] (X1) -- (X0);
    \draw [->] (A2) -- (A1);
    \draw [->] (A1) -- (A0);

    \draw [->] (V2) -- (Y2);
    \draw [->] (V2) -- (Y1);
    \draw [->] (V1) -- (Y1);
    \draw [->] (V1) -- (Y0);
    \draw [->] (V0) -- (Y0);

\end{tikzpicture}
\caption{A compatible FT-DAG and the minimal adjustment set that satisfies Condition~\ref{cond1}, item~\ref{adjcond1c} of the SCG-back-door criteria.}
\label{fig:incompletness:DAG2}
\end{subfigure}
\hfill
\begin{subfigure}[t]{.3\linewidth}
\centering
\begin{tikzpicture}

  \node (U2) {$U_{t-2}$};
  \node [draw=red] (U1) [right of= U2] {$U_{t-1}$};
  \node [draw=red] (U0) [right of= U1] {$U_{t}$};

  \node (A2) [below of= U2]{$W_{t-2}$};
  \node (A1) [right of= A2] {$W_{t-1}$};
  \node (A0) [right of= A1] {$W_{t}$};

  \node (X2) [below of= A2] {$X_{t-2}$};
  \node [draw=red] (X1) [right of= X2] {$X_{t-1}$};
  \node [draw=blue] (X0) [right of= X1] {$X_{t}$};

  \node (Y2) [below of= X2] {$Y_{t-2}$};
  \node (Y1) [right of= Y2] {$Y_{t-1}$};
  \node [draw=blue] (Y0) [right of= Y1] {$Y_{t}$};

  \node (V2) [below of= Y2] {$R_{t-2}$};
  \node [draw=red] (V1) [right of= V2] {$R_{t-1}$};
  \node (V0) [right of= V1] {$R_{t}$};

    \draw [->] (U2) -- (A2);
    \draw [->] (U2) -- (A1);
    \draw [->] (U1) -- (A1);
    \draw [->] (U1) -- (A0);
    \draw [->] (U0) -- (A0);

    \draw[->, bend right=20] (U2) to (V2);
    \draw[->] (U2) to (V1);
    \draw[->, bend right=20] (U1) to (V1);
    \draw[->] (U1) to (V0);
    \draw[->, bend right=20] (U0) to (V0);

    \draw [->] (A2) -- (X2);
    \draw [->] (A2) -- (X1);
    \draw [->] (A1) -- (X1);
    \draw [->] (A1) -- (X0);
    \draw [->] (A0) -- (X0);

    \draw [->] (X2) -- (Y2);
    \draw [->] (X2) -- (Y1);
    \draw [->] (X1) -- (Y1);
    \draw [->] (X1) -- (Y0);    
    \draw [->] (X0) -- (Y0);

    \draw [->] (X2) -- (X1);
    \draw [->] (X1) -- (X0);
    \draw [->] (A2) -- (A1);
    \draw [->] (A1) -- (A0);

    \draw [->] (V2) -- (Y2);
    \draw [->] (V2) -- (Y1);
    \draw [->] (V1) -- (Y1);
    \draw [->] (V1) -- (Y0);
    \draw [->] (V0) -- (Y0);

\end{tikzpicture}
\caption{A compatible FT-DAG and another adjustment set that satisfies the back-door~\citep{Pearl_1993BackDoor}.}
\label{fig:incompletness:DAG4}
\end{subfigure}

\caption{Example of an SCG and one compatible FT-DAG with three adjustment sets that satisfy condition~\ref{cond1} of the SCG-back-door criteria and one adjustment set that satisfies the back-door criteria in \cite{Pearl_1993BackDoor}.}
\label{fig:incompletness}
\end{figure}

\section{Finding the quasi-optimal adjustment set for the causal effect in SCGs}
\label{sec:optimal}

Firstly, let us recall the definition of an optimal adjustment set in a DAG \citep[Definition 3.12]{Henckel_2022} and adapt it in the case of FT-DAGs. 
\begin{definition}[Optimal adjustment set in an FT-DAG]
\label{def:opt-ftdag}
    Let $\FTCG$ be a FT-DAG and $(X_{t-\gamma},Y_{t})$ be two temporal variables of interest such that $X_{t-\gamma} \in \ancestors{Y_{t}}{\FTCG}$.
    We define the optimal adjustment set of the pair $(X_{t-\gamma},Y_{t})$ in the FT-DAG $\FTCG$ as follows:
    \begin{itemize}
        \item $\mathbb{P} = \parents{\cn{X_{t-\gamma}}{Y_t}{\FTCG}\backslash\{X\}}{\FTCG}$, and
        \item $\mathbb{D} = \descendants{X_{t-\gamma}}{\FTCG}$, and
        \item $\opt{X_{t-\gamma}}{Y_t}{\FTCG} = \mathbb{P} \backslash \mathbb{D}$.
    \end{itemize}
\end{definition}

Upon its introduction, $\opt{X_{t-\gamma}}{Y_t}{\FTCG}$ was shown to yield the smallest asymptotic variance when estimating causal effects under two conditions:  the SCM is linear and satisfies Assumption~\ref{ass:causal sufficiency}, and  the causal effect is estimated using the coefficient of the treatment in an ordinary least squares regression of the outcome on the treatment and a valid adjustment set~\cite{Henckel_2022}.
Subsequent work~\cite{Rotnitzky_2020} extended this optimality result to non-parametric SCMs\textemdash where no specific assumptions are made about the conditional distributions of nodes given their parents\textemdash and to non-parametric estimators, which do not rely on the conditional independencies encoded in the SCM.

However, defining optimality in SCGs is more nuanced than in FT-DAG. We consider that a sound definition of optimality should be that a set $\opt{X_{t-\gamma}}{Y_t}{\SCG}$ is optimal in an SCG $\SCG$ if  for all adjustment sets $\mathbb{Z}$ that are valid in any FT-DAG compatible with  $\SCG$ and for every distribution $\Pr$ compatible with any FT-DAG compatible with $\SCG$,
$$\var{\Pr}{\estimator{X_{t-\gamma}}{Y_t}{\opt{X_{t-\gamma}}{Y_t}{\SCG}}} \le\var{\Pr}{\estimator{X_{t-\gamma}}{Y_t}{\mathbb{Z}}}$$
where $\var{\Pr}{\cdot}$ represents the asymptotical variance under distribution $\Pr$, $\causaleffect{X_{t-\gamma}}{Y_t}{\mathbb{Z}} = \sum_{\mathbb{z}} E(Y_t\mid X_{t-\gamma} = x, \mathbb{Z}=\mathbb{z}) -E(Y_t\mid X_{t-\gamma} = x', \mathbb{Z}=\mathbb{z})$ and $\estimator{X_{t-\gamma}}{Y_t}{\mathbb{Z}}$ is a consistent estimator of $\causaleffect{X_{t-\gamma}}{Y_t}{\mathbb{Z}}$ (such as ordinary least squares regression) used to estimate the micro-level causal effect of $X_{t-\gamma}$ on $Y_t$ using the adjustment set $\mathbb{Z}$.
The problem of identifying an optimal adjustment set in SCGs proves to be highly challenging. In fact, if we limit ourselves to the results provided in \citep{Henckel_2022}, it becomes impossible to fully solve. More precisely, \citep[Theorem 1]{Henckel_2022} establishes an inequality that can be used to compare two adjustment sets and determine which one is more optimal\textemdash\ie, which yields a lower asymptotic variance\textemdash when the sets satisfy different graphical conditions. However, this result is not complete: there exist pairs of sets for which the theorem does not allow a clear comparison in terms of optimality (that said, their approach does guarantee the identification of at least one set that is optimal relative to all others).

In this work, instead of focusing on finding optimal adjustment sets in SCGs, we will instead focus on finding quasi-optimal sets (quasi since we limit ourselves to the results given in~\cite{Henckel_2022}). We say that a set $\qopt{X_{t-\gamma}}{Y_t}{\SCG}$ is quasi-optimal in an SCG if:
\begin{enumerate}[label=(\textbf{\roman*})]
    \item\label{item:1} $\qopt{X_{t-\gamma}}{Y_t}{\SCG}$ is a valid adjustment set for the SCG;
    \item\label{item:2}  there exists at least one FT-DAG $\FTCG$ compatible with $\SCG$ such that for all adjustment sets $\mathbb{Z}$ that are valid in any FT-DAG compatible with  $\SCG$ and for every distribution $\Pr$   compatible $\FTCG$,
$$\var{\Pr}{\estimator{X_{t-\gamma}}{Y_t}{\qopt{X_{t-\gamma}}{Y_t}{\SCG}}} \le \var{\Pr}{\estimator{X_{t-\gamma}}{Y_t}{\mathbb{Z}}};$$
    \item\label{item:3} for every $\mathcal{G}$ compatible with $\mathcal{G}^s$,
    $\qopt{X_{t-\gamma}}{Y_t}{\SCG}$ is the \emph{smallest} among all valid adjustment for all FT-DAGs compatible with $\mathcal{G}^s$ that contains \emph{as much as possible} of the optimal adjustment set of every FT-DAG compatible with $\mathcal{G}^s$. 
\end{enumerate}
The item~\ref{item:1} mirrors the requirement for optimality in the standard FT-DAG setting. Items~\ref{item:2} and \ref{item:3} highlight the challenge of defining optimality in SCGs: since different FT-DAGs compatible with the same SCG may have different optimal sets, no single adjustment set can be optimal for all of them simultaneously. We argue that item~\ref{item:2} is essential for any reasonable notion of optimality in SCGs, while item~\ref{item:3} serves as a strong sufficiency criterion. 
While we do not present a formal proof, empirical results indicate that adjustment sets containing \emph{as much as possible} of the optimal set generally produce more accurate estimates than those that do not. This observation motivates the intuition behind item~\ref{item:3} and guides the design of the quasi-optimal adjustment set, which seeks to balance robustness and performance in the face of SCG ambiguity.

Next, we can characterize the optimal adjustment set given only the SCG.
\begin{definition}[Quasi-optimal adjustment set in an SCG]
\label{def:q-opt-scg}
    Let $\SCG$ be an SCG and $(X_{t-\gamma},Y_{t})$ be two temporal variables of interest such that $X \in \ancestors{Y}{\SCG}$.
    We define the quasi-optimal adjustment set of the pair $(X_{t-\gamma},Y_{t})$ in the SCG $\SCG$ as         $\qopt{X_{t-\gamma}}{Y_{t}}{\SCG} = \mathbb{P} \backslash \mathbb{D}$, where $\mathbb{D} = \posdescendants{X_{t-\gamma}}{\SCG}$, and :
    \begin{itemize}
        \item If Condition~\ref{cond1} is verified and $\cycles{X}{\SCG} = \emptyset$ and $\gamma>0$:\\ $\mathbb{P} = \temp{\parents{\ecn{X}{Y}{\SCG}}{\SCG}}{t}{t-\gamma-\gamma_{max}}$, and
        \item If Condition~\ref{cond1} is verified and $\gamma=0$:\\ $\mathbb{P} = \temp{\parents{\cn{X}{Y}{\SCG}}{\SCG}}{t}{t-\gamma-\gamma_{max}} \cup \temp{\parents{\ecnbd{X}{Y}{\SCG}{\emptyset}}{\SCG}}{t}{t-\gamma-\gamma_{max}}$, and
        \item If Condition~\ref{cond1} is verified and $\cycles{X}{\SCG} \neq \emptyset$ and $\gamma>0$:\\ $\mathbb{P} = \temp{\parents{X}{\SCG}}{t}{t-\gamma-\gamma_{max}+1} \cup \temp{\parents{\ecn{X}{Y}{\SCG}}{\SCG}}{t}{t-\gamma-\gamma_{max}}$, and
        \item If Condition~\ref{cond2} is verified: \\
        $\mathbb{P} = \temp{\parents{\cn{X}{Y}{\SCG}}{\SCG}}{t}{t-\gamma-\gamma_{max}} \cup \temp{\parents{\ecnbd{X}{Y}{\SCG}{\emptyset}}{\SCG}}{t}{t-\gamma-\gamma_{max}}$, and
        \item If Condition~\ref{cond3} is verified: \\
        $\mathbb{P} = \temp{\parents{Y}{\SCG}}{t}{t-\gamma-\gamma_{max}} \cup \temp{\parents{X}{\SCG}}{t}{t-\gamma_{max}}$.
    \end{itemize}
\end{definition}


In the following,  we  give three propositions that show that the  set given in Definition~\ref{def:q-opt-scg} satisfies the three items given above, therefore corresponds to the actual quasi-optimal set.

\begin{restatable}{proposition}{propositionOptimalIsGivenBySCGBackDoor}
\label{prop:OptimalIsGivenBySCGBackDoor}
    Let $\SCG$ be an SCG and $(X_{t-\gamma},Y_{t})$ be two temporal variables of interest such that $X \in \ancestors{Y}{\SCG}$ and such that the micro causal effect from $X_{t-\gamma}$ to $Y_{t}$ is identifiable by Theorem~\ref{theorem:identifiability_total_effect_SCG}.
    The set $\qopt{X_{t-\gamma}}{Y_{t}}{\SCG}$ satisfies the SCG-back-door criterion.
\end{restatable}

\begin{restatable}{proposition}{propositionOptimalInOneFTDAG}
\label{prop:OptimalInOneFTDAG}
    Let $\SCG$ be an SCG and $(X_{t-\gamma},Y_{t})$ be two temporal variables of interest such that $X \in \ancestors{Y}{\SCG}$ and such that the micro causal effect from $X_{t-\gamma}$ to $Y_{t}$ is identifiable by the SCG-back-door criterion.
    Suppose there exists a compatible FT-DAG $\FTCG$ such that $\opt{X_{t-\gamma}}{Y_t}{\FTCG}$ is a valid adjustment set in every compatible FT-DAG.
    There exist a compatible FT-DAG $\FTCG$ in which $\opt{X_{t-\gamma}}{Y_t}{\FTCG}=\qopt{X_{t-\gamma}}{Y_{t}}{\SCG}$.
\end{restatable}

\begin{restatable}{proposition}{propositionIncludeAllOptimals}
\label{prop:IncludeAllOptimals}
    Let $\SCG$ be an SCG and $(X_{t-\gamma},Y_{t})$ be two temporal variables of interest such that $X \in \ancestors{Y}{\SCG}$ and such that the micro causal effect from $X_{t-\gamma}$ to $Y_{t}$ is identifiable by the SCG-back-door criterion.
    For every compatible FT-DAG $\FTCG$, $\opt{X_{t-\gamma}}{Y_t}{\FTCG} \backslash \posdescendants{X_{t-\gamma}}{\SCG}\subseteq\qopt{X_{t-\gamma}}{Y_{t}}{\SCG}$.
\end{restatable}

These propositions suggest that the quasi-optimal set can be interpreted as the union of all optimal adjustment sets across FT-DAGs compatible with the SCG, excluding all possible descendants of $X_{t-\gamma}$. This observation is formalized by the following corollary:

\begin{restatable}{corollary}{corollaryUnionAllOptimals}
\label{coro:UnionAllOptimals}
    Let $\SCG$ be an SCG and $(X_{t-\gamma},Y_{t})$ be two temporal variables of interest such that $X \in \ancestors{Y}{\SCG}$ and such that the causal effect from $X_{t-\gamma}$ to $Y_{t}$ is identifiable.
    Suppose there exists a compatible FT-DAG $\FTCG$ such that $\opt{X_{t-\gamma}}{Y_t}{\FTCG}$ is a valid adjustment set in every compatible FT-DAG.
    Let $\fulltime{\SCG}$ be the set of FT-DAGs compatible with $\SCG$.
    The following equality holds: $$\qopt{X_{t-\gamma}}{Y_{t}}{\SCG} = \bigcup\limits_{\FTCG \in \fulltime{\SCG}}\opt{X_{t-\gamma}}{Y_{t}}{\FTCG}\backslash \posdescendants{X_{t-\gamma}}{\SCG}$$
\end{restatable}

Proposition~\ref{prop:IncludeAllOptimals} (or Corollary~\ref{coro:UnionAllOptimals}) might initially appear overly weak, as it suggests that \emph{only} a subset of the optimal adjustment set is included within the quasi-optimal set for any FT-DAG compatible with the SCG. This could mistakenly lead one to believe that a valid quasi-optimal set must contain the optimal adjustment set of every such FT-DAG. However, this is not feasible in general: the union of all optimal sets across compatible FT-DAGs is not necessarily a valid adjustment set in each of them, \ie it does not verify the generalized adjustment criterion in each FT-DAG. In other words, this union may fail to form a valid adjustment set in the SCG.
This is illustrated in Figure~\ref{fig:pb_optimality}, where $U$ appears as both a parent and a child of the outcome in the SCG shown in Figure~\ref{fig:pb_optimality:SCG}. As a result, $\opt{X_{t-}}{Y_t}{\FTCG} \not\subseteq\qopt{X_{t}}{Y_{t}}{\SCG}$ in such cases. Indeed, among the FT-DAGs compatible with this SCG, one can construct a FT-DAG (Figure~\ref{fig:pb_optimality:DAG2}) in which $U_t$ is a descendant of the outcome $Y_t$. In that FT-DAG, any adjustment set including $U_t$ would be invalid, since adjusting for a descendant of the outcome violates the back-door criterion. Therefore, although $U_t$ may belong to an optimal adjustment set in another compatible FT-DAG (e.g., Figure~\ref{fig:pb_optimality:DAG1}), it cannot be included in a quasi-optimal set that is valid across all compatible FT-DAGs.

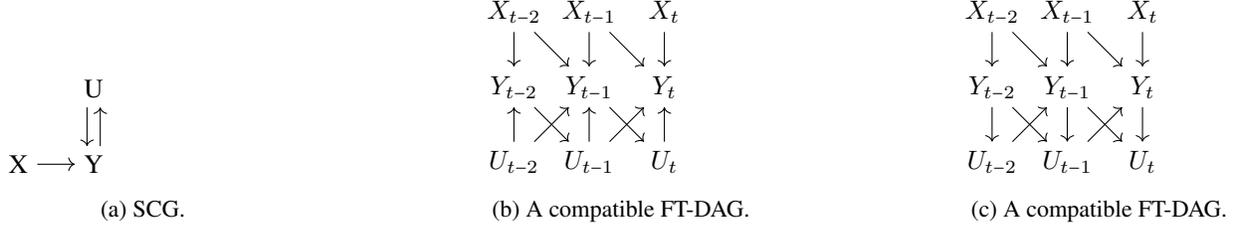
\begin{figure}[t]
    \label{fig:pb_optimality}
    \centering

\begin{subfigure}[t]{.23\linewidth}
\begin{tikzpicture}
    
    \node (Z) {U};
    \node (Y) [below of= U] {Y};
    \node (X) [left of= Y] {X};

    \draw [->] (X) -- (Y);
    \begin{scope}[transform canvas={xshift=-.25em}]
    \draw [->] (U) -- (Y);
    \end{scope}
    \begin{scope}[transform canvas={xshift=.25em}]
    \draw [<-] (U) -- (Y);
    \end{scope}

\end{tikzpicture}
\caption{SCG.}
\label{fig:pb_optimality:SCG}
\end{subfigure}
\hfill
\begin{subfigure}[t]{.23\linewidth}
\begin{tikzpicture}

  \node (X2) {$X_{t-2}$};
  \node (X1) [right of= X2] {$X_{t-1}$};
  \node (X0) [right of= X1] {$X_{t}$};

  \node (Y2) [below of= X2] {$Y_{t-2}$};
  \node (Y1) [right of= Y2] {$Y_{t-1}$};
  \node (Y0) [right of= Y1] {$Y_{t}$};

  \node (U2) [below of= Y2] {$U_{t-2}$};
  \node (U1) [right of= U2] {$U_{t-1}$};
  \node (U0) [right of= U1] {$U_{t}$};

    \draw [->] (X2) -- (Y2);
    \draw [->] (X2) -- (Y1);
    \draw [->] (X1) -- (Y1);
    \draw [->] (X1) -- (Y0);    
    \draw [->] (X0) -- (Y0);

    \draw [->] (U2) -- (Y2);
    \draw [->] (U2) -- (Y1);
    \draw [->] (U1) -- (Y1);
    \draw [->] (U1) -- (Y0);
    \draw [->] (U0) -- (Y0);

    \draw [->] (Y2) -- (U1);
    \draw [->] (Y1) -- (U0);

\end{tikzpicture}
\caption{A compatible FT-DAG.}
\label{fig:pb_optimality:DAG1}
\end{subfigure}
\hfill
\begin{subfigure}[t]{.23\linewidth}
\begin{tikzpicture}

  \node (X2) {$X_{t-2}$};
  \node (X1) [right of= X2] {$X_{t-1}$};
  \node (X0) [right of= X1] {$X_{t}$};

  \node (Y2) [below of= X2] {$Y_{t-2}$};
  \node (Y1) [right of= Y2] {$Y_{t-1}$};
  \node (Y0) [right of= Y1] {$Y_{t}$};

  \node (U2) [below of= Y2] {$U_{t-2}$};
  \node (U1) [right of= U2] {$U_{t-1}$};
  \node (U0) [right of= U1] {$U_{t}$};

    \draw [->] (X2) -- (Y2);
    \draw [->] (X2) -- (Y1);
    \draw [->] (X1) -- (Y1);
    \draw [->] (X1) -- (Y0);    
    \draw [->] (X0) -- (Y0);

    \draw [->] (Y2) -- (U2);
    \draw [->] (Y2) -- (U1);
    \draw [->] (Y1) -- (U1);
    \draw [->] (Y1) -- (U0);
    \draw [->] (Y0) -- (U0);

    \draw [->] (U2) -- (Y1);
    \draw [->] (U1) -- (Y0);

\end{tikzpicture}
\caption{A compatible FT-DAG.}
\label{fig:pb_optimality:DAG2}
\end{subfigure}

\caption{Example of an SCG and two compatible FT-DAGs illustrating that the optimal adjustment set is not necessarily contained within the quasi-optimal set.}

\label{fig:pb_optimality}

\end{figure}

Notice that what we defined as quasi-optimal set is a subset of $\mathbb{A}^1$ and $\mathbb{A}^2$, as we recall that $\mathbb{A}^1$ is the set containing all variables that are not descendants of the treatment and $\mathbb{A}^2$ is the set containing all ancestors of the treatment and the outcome.
Interestingly, the optimal set is a subset of $\mathbb{A}^1$ and $\mathbb{A}^2$ if and only if the optimal set is a subset of the quasi-optimal (in cases where the optimal set contains no descendants of the treatment).
The following corollary relates the quasi-optimal set with the sets $\mathbb{A}^1$ and $\mathbb{A}^2$ introduced in \cite{Assaad_2024}: 

\begin{restatable}{corollary}{corollaryQoptVsAs}
\label{coro:QoptVsAs}
    Let $\SCG$ be an SCG and $(X_{t-\gamma},Y_{t})$ be two temporal variables of interest such that $X \in \ancestors{Y}{\SCG}$ and such that the causal effect from $X_{t-\gamma}$ to $Y_{t}$ is identifiable by Theorem~\ref{theorem:identifiability_total_effect_SCG}. For any distribution $\Pr$ compatible with any FT-DAG compatible with $\mathcal{G}^s$. 
    Under the same assumptions (briefly discussed above and described in details in \cite{Henckel_2022,Rotnitzky_2020}) required for the optimality of the set $\opt{X_{t-\gamma}}{Y_t}{\FTCG}$, we have the following:
$$\var{\Pr}{\estimator{X_{t-\gamma}}{Y_t}{\qopt{X_{t-\gamma}}{Y_t}{\mathcal{G}^s}}} \le \var{\Pr}{\estimator{X_{t-\gamma}}{Y_t}{\mathbb{A}^1}} \text{, and }$$
$$\var{\Pr}{\estimator{X_{t-\gamma}}{Y_t}{\qopt{X_{t-\gamma}}{Y_t}{\mathcal{G}^s}}} \le \var{\Pr}{\estimator{X_{t-\gamma}}{Y_t}{\mathbb{A}^2}}.$$    
\end{restatable}


The results of this section indicate that the set $\qopt{X_{t-\gamma}}{Y_{t}}{\SCG}$ is optimal in at least one of the FT-DAGs compatible with the SCG and outperforms several other adjustment sets (\eg, $\mathbb{A}^1$ and $\mathbb{A}^2$).  These findings have direct practical implications: they provide guidance for selecting adjustment sets when working with multiple time series (such as cohort data) generated from a DTDSCM, in cases where only the SCG is known. We further conjecture that these insights may extend to classical time series settings under Assumption~\ref{ass:ctt}, offering a principled approach to adjustment set selection in that context as well.


\section{Discussion}
\label{sec:discussion}

\begin{figure}[t]
\centering

\tikzset{every picture/.style={scale=1, every node/.style={scale=0.8}}}

\begin{subfigure}[b]{\textwidth}
\centering
\begin{tikzpicture}
\begin{axis}[
    width=\textwidth,
    height=2cm,
    ylabel={$X$},
    ymin=0, ymax=5,
    xmin=0, xmax=10,
    xtick={0,...,10},
    ytick=\empty,
    grid=major,
    domain=0:10,
    samples=11
]
\addplot[only marks, mark=*, mark size=2pt, color=blue] 
    coordinates {(0,1) (1,2) (2,1.5) (3,3) (4,2.2) (5,3.5) (6,2.8) (7,4) (8,3.5) (9,4.5) (10,3.8)};
\end{axis}

\end{tikzpicture}

\begin{tikzpicture}
\begin{axis}[
    width=\textwidth,
    height=2cm,
    xlabel={Time (seconds)},
    ylabel={$Y$},
    ymin=0, ymax=5,
    xmin=0, xmax=10,
    xtick={0,...,10},
    ytick=\empty,
    grid=major,
    domain=0:10,
    samples=11
]
\addplot[only marks, mark=*, mark size=2pt, color=red] 
    coordinates {(0,3) (1,1.5) (2,2.4) (3,2.8) (4,1.9) (5,2) (6,4.6) (7,3.8) (8,2) (9,4.5) (10,2.5)};
\end{axis}

\end{tikzpicture}
\caption{Data generated every second by a DTDSCM.}
\end{subfigure}

\begin{subfigure}[b]{\textwidth}
\centering
\begin{tikzpicture}
\begin{axis}[
    width=\textwidth,
    height=2cm,
    ylabel={$X$},
    ymin=0, ymax=5,
    xmin=0, xmax=10,
    xtick={0,...,10},
    ytick=\empty,
    grid=major,
    domain=0:10,
    samples=11
]
\addplot[only marks, mark=*, mark size=2pt, color=blue] 
    coordinates {(0,1)  (2,1.5)  (4,2.2)  (6,2.8) (8,3.5)  (10,3.8)};
\end{axis}

\end{tikzpicture}

\begin{tikzpicture}
\begin{axis}[
    width=\textwidth,
    height=2cm,
    xlabel={Time (seconds)},
    ylabel={$Y$},
    ymin=0, ymax=5,
    xmin=0, xmax=10,
    xtick={0,...,10},
    ytick=\empty,
    grid=major,
    domain=0:10,
    samples=11
]
\addplot[only marks, mark=*, mark size=2pt, color=red] 
    coordinates {(0,3)  (2,2.4)  (4,1.9)  (6,4.6) (8,2) (10,2.5)};
\end{axis}

\end{tikzpicture}
\caption{Data sampled data every two seconds.}
\end{subfigure}

\begin{subfigure}[b]{\textwidth}
\centering

\begin{tikzpicture}
\tikzset{hidden/.append style={draw=black, circle}}

  \node (X10)  {$X_{0}$};
  \node[hidden] (X9) [right of= X10] {$X_{1}$};
  \node (X8) [right of= X9] {$X_{2}$};
  \node[hidden] (X7) [right of= X8] {$X_{3}$};
  \node (X6) [right of= X7] {$X_{4}$};
  \node[hidden] (X5) [right of= X6] {$X_{5}$};
  \node (X4) [right of= X5] {$X_{6}$};
  \node[hidden] (X3) [right of= X4] {$X_{7}$};
  \node (X2) [right of= X3] {$X_{8}$};
  \node[hidden] (X1) [right of= X2] {$X_{9}$};
  \node (X0) [right of= X1] {$X_{10}$};

  \node (Y10) [below of= X10] {$Y_{0}$};
  \node[hidden] (Y9) [below of= X9] {$Y_{1}$};
  \node (Y8) [below of= X8] {$Y_{2}$};
  \node[hidden] (Y7) [below of= X7] {$Y_{3}$};
  \node (Y6) [below of= X6] {$Y_{4}$};
  \node[hidden] (Y5) [below of= X5] {$Y_{5}$};
  \node (Y4) [below of= X4] {$Y_{6}$};
  \node[hidden] (Y3) [below of= X3] {$Y_{7}$};
  \node (Y2) [below of= X2] {$Y_{8}$};
  \node[hidden] (Y1) [right of= Y2] {$Y_{9}$};
  \node (Y0) [right of= Y1] {$Y_{10}$};
    
    \draw [->] (X10) -- (X9);
    \draw [->] (X9) -- (X8);
    \draw [->] (X8) -- (X7);
    \draw [->] (X7) -- (X6);
    \draw [->] (X6) -- (X5);
    \draw [->] (X5) -- (X4);
    \draw [->] (X4) -- (X3);
    \draw [->] (X3) -- (X2);
    \draw [->] (X2) -- (X1);
    \draw [->] (X1) -- (X0);

    \draw [->] (Y10) -- (Y9);
    \draw [->] (Y9) -- (Y8);
    \draw [->] (Y8) -- (Y7);
    \draw [->] (Y7) -- (Y6);
    \draw [->] (Y6) -- (Y5);
    \draw [->] (Y5) -- (Y4);
    \draw [->] (Y4) -- (Y3);
    \draw [->] (Y3) -- (Y2);
    \draw [->] (Y2) -- (Y1);
    \draw [->] (Y1) -- (Y0);


    \draw [->] (X10) -- (Y8);
    \draw [->] (X9) -- (Y7);
    \draw [->] (X8) -- (Y6);
    \draw [->] (X7) -- (Y5);
    \draw [->] (X6) -- (Y4);
    \draw [->] (X5) -- (Y3);
    \draw [->] (X4) -- (Y2);
    \draw [->] (X3) -- (Y1);
    \draw [->] (X2) -- (Y0);

\end{tikzpicture}
\caption{FT-DAG induced by the DTDSCM where circles represents unmeasured variables.}
\end{subfigure}

\caption{Comparison of generated vs sampled time series data. The upper plot in (a) shows values generated every second. The lower plot in (b) shows the same data sampled every two seconds. The FT-DAG induced by the generative process is given in (c).}
\label{fig:samplingVsgeneration}
\end{figure}
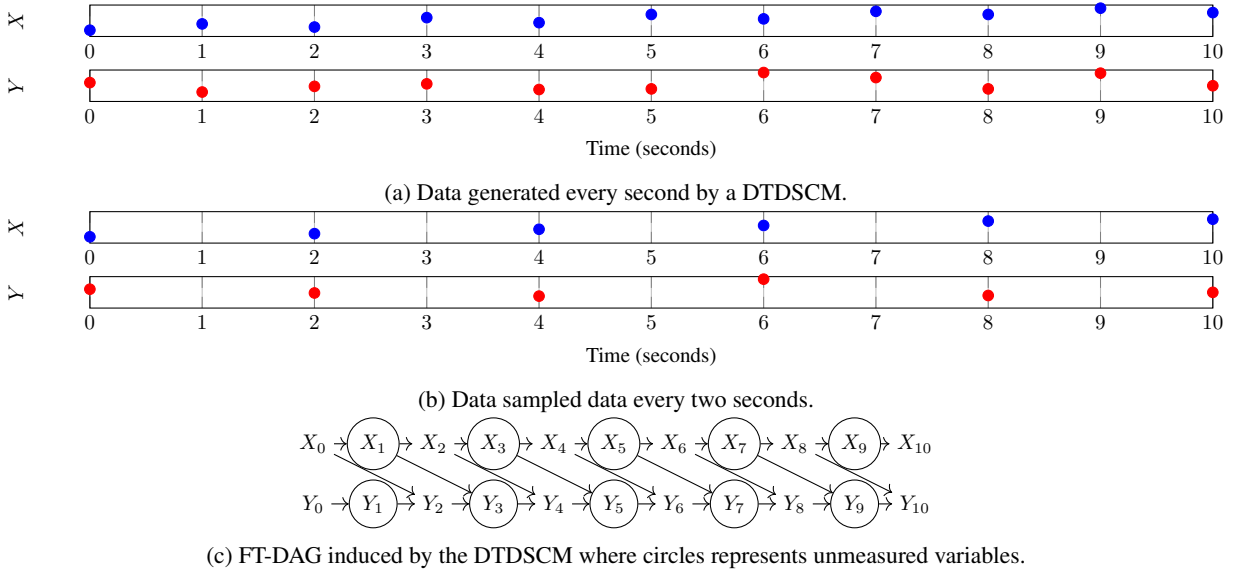


It is important to emphasize that Conditions~\ref{cond1}, \ref{cond2} and \ref{cond3} of Theorem~\ref{theorem:identifiability_total_effect_SCG} rely heavily on the assumptions of a DTDSCM, in particular that the data collection frequency perfectly matches the data generation process.
If this assumption is violated, there is a risk of backdoor paths between $X_{t-\gamma}$ and $Y_t$ consisting entirely of unmeasured variables, which prevents valid adjustment. For instance, as illustrated in Figure~\ref{fig:samplingVsgeneration}, if the underlying system generates values at every second\textemdash  $\mathbb{v}_{t_0}, \mathbb{v}_{t_0 +1}, \mathbb{v}_{t_0 +2}, \cdots \mathbb{v}_{t_{max}-2}, \mathbb{v}_{t_{max}-1},\mathbb{v}_{t_{max}}$\textemdash  but observations are recorded only every two seconds\textemdash $\mathbb{v}_{t_0}, \mathbb{v}_{t_0 +2} \cdots , \mathbb{v}_{t_{max}-2}, \mathbb{v}_{t_{max}}$\textemdash then it is not possible to identify the micro causal effect of $X_{4}$ on $Y_6$ since the backdoor path consisting of the nodes $X_3$ and $Y_3$, which are both unmeasured, cannot be blocked.
In addition, discrepancies between the data-generating process and the data collection frequency may lead to misinterpretation in Condition~\ref{cond3}. 
For example, in scenarios like the one illustrated in Figure~\ref{fig:samplingVsgeneration},
a modeler might mistakenly set $\gamma$ to $1$ to reflect the difference in observed time steps (which is actually $2$).
This could lead to the erroneous belief that Condition~\ref{cond3} applies for the micro causal effect of interest $\probac{y_t}{\interv{x_{t-1}}}$, when in reality the actual  micro causal effect of interest corresponds to $\probac{y_t}{\interv{x_{t-2}}}$.
Such mismatches highlight the need for caution when applying Condition~\ref{cond3} or any lag-based (where the lag is set to a value greater than zero) identifiability conditions to discretely observed time series. One might assume that Condition~\ref{cond2} suffers from the same limitation as Condition~\ref{cond3}, since it also relies on a specific lag\textemdash namely, $\gamma = 0$. However, having a lag of zero corresponds to an instantaneous relationship, which remains unaffected by differences between the data generation rate and the sampling rate. Thus, Condition~\ref{cond2} is not subject to the same limitations. 

As this paper focuses on time series, it is necessary to define a temporal boundary, since measurements begin at a finite point in time. We choose to set this limit at $t - \gamma - \gamma_{max}$, which places us sufficiently far in the past to ensure that all relevant back-door paths are blocked. This choice also reflects real-world constraints, where infinite-length measurements are not feasible. Moreover, by selecting this time limit, we ensure that no more optimal adjustment set would be discovered by going even further back in time\textemdash any variables beyond this boundary would either be redundant or irrelevant for optimality.

We evaluated our results using simulated multivariate time series data. However, it is important to note that the theoretical guarantees underlying our approach extend beyond this setting. Specifically, they apply not only to single multivariate time series but also to collections of multiple multivariate time series, as commonly encountered in cohort studies. As such, we expect similar empirical performance to hold in those broader contexts as well if all assumptions are met. However, it is important to emphasize that our approach is applicable to cohort studies only when the necessary assumptions hold and the data is complete, without missing values or irregularities in the timing of measurements.
Furthermore, while our analysis relies on the standard positivity assumption, this condition can be relaxed. In particular, it suffices that for all values $\mathbb{z}$ of the adjustment set $\mathbb{Z}$, either $\mathbb{P}(\mathbb{z}) > 0$ or $\mathbb{P}(x_{t-\gamma} \mid \mathbb{z}) > 0$, as shown by \cite{Hwang_2024}. This relaxed condition broadens the applicability of our results to more realistic scenarios where strict positivity may not hold globally.

In many applications, it is falsely assumed that the minimal adjustment set is also the most optimal. This misconception has been challenged, for example, in \cite{Henckel_2022} and is further examined in this paper. While minimal adjustment sets are often preferred for their simplicity and efficiency, they do not always lead to the lowest estimation variance or the most robust causal effect estimation. In particular, including additional variables \textemdash when carefully chosen \textemdash can improve statistical efficiency in the presence of unmeasured confounding structures. This highlights the importance of going beyond minimality and taking into account statistical efficiency when choosing among valid adjustment sets.

\section{Algorithmic validity}
\label{sec:exp}

We began by randomly generating 1500 graphs—both cyclic and acyclic—each containing five or six nodes. These graphs served as our SCGs. If at least one of the identifiability conditions in Theorem~\ref{theorem:identifiability_total_effect_SCG} was satisfied, we then generated all compatible FT-DAGs by setting $\gamma_{max}=1$. Since the number of compatible FT-DAGs grows exponentially with the number of cycles in the SCG, we restricted our analysis to SCGs that generate at most $50$ compatible densest FT-DAGs, in order to keep computation time manageable. For each selected SCG, we conducted a two-step comparison. First, we identified all valid adjustment sets derived from the SCG, as determined by the SCG-back-door criterion. Then, we searched for all adjustment sets that satisfie the back-door criterion in all the corresponding FT-DAGs. We will call the former  sets as the SCG-back-door sets and the latter sets as the common-back-door sets.
We compared the two sets to assess the validity of our criterion. 

As a result, we observed that all SCG-back-door sets are also common-back-door sets. This empirical comparison allowed us to evaluate a key theoretical property of our method: soundness, meaning that every adjustment set selected by our criterion is guaranteed to be valid across all compatible FT-DAGs. This supports the reliability and robustness of the SCG-back-door criterion.

However, we also observed that not every common back-door set qualifies as an SCG-back-door set. This underscores the incompleteness of the SCG-back-door criterion. It is important to note, however, that our objective was not to establish a complete criterion, but rather to develop a sufficiently expressive one (similarly to the back-door criterion which is itself not complete)\textemdash capable of identifying multiple valid adjustment sets, among which a quasi-optimal set can be found.

\section{Conclusion}
\label{sec:conclusion}

In this paper, we addressed the problem of identifying micro causal effects from observational data using SCGs. Building on prior work, we proposed a simplified and interpretable reformulation of the identifiability by adjustment conditions in SCGs originally introduced in \cite{Assaad_2024}. Our formulation avoids the enumeration of all active paths and relies solely on kinship-based graph properties, making it more tractable in practice.
Furthermore, we introduced the SCG-back-door criterion, a graphical test that extends the classical back-door criterion to the SCG setting, providing a sound and efficient method to enumerate valid adjustment sets whenever the micro causal effect is identifiable by adjustment. In addition, we characterized the quasi-optimal adjustment set among all sets satisfying the SCG-back-door criterion, thus offering a principled way to improve estimation efficiency.
Our results not only provide theoretical insights into causal identifiability by adjustment in SCGs but also offer practical tools for applied researchers working with time series data or cohorts where only high-level causal abstractions are available. 

Applying the theoretical findings developed in this work to real-world epidemiological datasets would be a valuable next step to assess their practical impact.
Additionally, future research could explore characterizing the optimal set from SCGs when there is hidden confounding using the recently proposed SCG-front-door criterion~\citep{Assaad_2025}, or investigate how the classical single-door criterion might be extended to SCGs, building on the identifiability conditions introduced in \cite{Ferreira_2024} for micro direct effect identification in linear SCMs. 
It would also be interesting to investigate whether the forbidden projection technique introduced in \cite{Witte_2020} and which simplifies DAGs while preserving all information relevant to identifying optimal adjustment sets, can be extended to SCGs. Another promising direction is to focus on cohort studies with imperfect data, characterized by irregular measurement timings and missing values. Since current criteria assume well-structured and complete data, it would be valuable to develop approaches that can accommodate the complexities and imperfections commonly encountered in real-world datasets.
However, since this work builds on the results of \cite{Henckel_2022}, which are not yet complete, the question of whether a more optimal adjustment set can be identified within the SCG remains open. We leave this investigation for future work.

\subsection*{Contribution}
This work was carried out during the internship of I.B., who contributed to the implementation, experimental evaluation, and validation of the theoretical results at the algorithmic level. All authors contributed to the conceptualization of the study. S.F. developed most of the theoretical proofs, and S.F. and C.A. jointly supervised the internship and the research. C.A. provided the funding for this study. All authors reviewed the manuscript. S.F. is designated as the corresponding author of this paper.

\subsection*{Acknowledgements}
We thank Timothée Loranchet for providing part of the code used to generate the simulated data, as well as Chantal Guihenneuc and David Hajage for their valuable comments on the manuscript.
This work was supported by the CIPHOD project (ANR-23-CPJ1-0212-01). 

\bibliography{references}

\newpage
\appendix
\section{Appendix}

\theoremIdentifiabilityTotalEffect*

\begin{proof}
    This theorem is a reformulation of Theorem 2 - Version 2 of \cite{Assaad_2024} which gives necessary and sufficient conditions for identification of a micro causal effect from a SCG.
    Firstly, they show that if $X\notin\ancestors{Y}{\SCG}$ then it is identifiable since $\probac{y_t}{\interv{x_{t-\gamma}}}=\proba{y_t}$.
    Secondly, suppose $X\in\ancestors{Y}{\SCG}$. If Condition 1 of Theorem 2 - Version 2 of \cite{Assaad_2024} is satisfied, then the absence of $\sigma$-active backdoor path in which every intermediate node is a descendant of $X$ together with the fact that $X\in\ancestors{Y}{\SCG}$ forbids the existence of a directed path from $Y$ to $X$ and thus guarantees $\cycles{X}{\SCG}^>=\emptyset$ or in other words $\scc{X}{\SCG}=\{X\}$. This corresponds to Condition~\ref{cond1} in Theorem~\ref{theorem:identifiability_total_effect_SCG}. Condition~\ref{cond1} in Theorem~\ref{theorem:identifiability_total_effect_SCG} also imply Condition 1 of Theorem 2 - Version 2 of \cite{Assaad_2024} as the absence of cycles other than self-loops on $X$ guarantees the absence of cycles containing $X$ and $Y$ and also guarantees the absence of a backdoor path in which the parent of $X$ is a descendant of $X$.
    Furthermore, if Condition 2 of Theorem 2 - Version 2 of \cite{Assaad_2024} is satisfied, then clearly Condition~\ref{cond2} of Theorem~\ref{theorem:identifiability_total_effect_SCG} is satisfied as having $V\in\ancestors{Y}{\SCG\backslash\{Y\}}\cap\scc{X}{\SCG}$ would force the existence of the $\sigma$-active backdoor path $\langle X \leftarrow \cdots \leftarrow V \rightarrow \cdots \rightarrow Y\rangle$. Condition~\ref{cond2} in Theorem~\ref{theorem:identifiability_total_effect_SCG} also implies Condition 2 of Theorem 2 - Version 2 of \cite{Assaad_2024} as the existence of a $\sigma$-active backdoor path $\pi=\langle X=V^0 \leftarrow V^1 \cdots V^n = Y\rangle$ in which every intermediate node is a descendant of $X$ together with the fact that $X\in\ancestors{Y}{\SCG}$ implies that either $Y\in\scc{X}{\SCG}$ or $\exists 1 \leq i \leq n-1$ such that $\langle V^{i-1} \leftarrow V^i \rightarrow V^{i+1}\rangle \subseteq \pi$ and $V^i \in \ancestors{Y}{\SCG} \cap \scc{X}{\SCG}$ which both contradict Condition~\ref{cond2} of Theorem~\ref{theorem:identifiability_total_effect_SCG}.
    Lastly, Condition 3 of Theorem 2 - Version 2 of \cite{Assaad_2024} is clearly equivalent to Condition~\ref{cond3} of Theorem~\ref{theorem:identifiability_total_effect_SCG}.

    In conclusion Theorem~\ref{theorem:identifiability_total_effect_SCG} and Theorem 2 of \cite{Assaad_2024} are equivalent.
\end{proof}

\theoremSCGbackdoor*

\begin{proof}
    Firstly, notice that $X\notin \ancestors{Y}{\SCG}$ and $\posdescendants{X_{t-\gamma}}{\SCG}\cap\mathbb{Z}=\emptyset$ then in every compatible FT-DAG $\FTCG$ one has $X_{t-\gamma}\notin\ancestors{Y_t}{\FTCG}$ and thus for every $\mathbb{Z}$ such that $\mathbb{Z}\cap\descendants{X_{t-\gamma}}{\FTCG}=\emptyset$, one has $\probac{y_t}{\interv{x_{t-\gamma}}} = \proba{y_t} = \sum_{\mathbb{z}}\probac{y_t}{\mathbb{z}}\proba{\mathbb{z}}$ using rule 2 of do-calculus and the law of total probabilities.
    
    Then assume that $X\in \ancestors{Y}{\SCG}$ and $\posdescendants{X_{t-\gamma}}{\SCG}\cap\mathbb{Z}=\emptyset$.
    \begin{itemize}
        \item Suppose Condition~\ref{cond1} is verified
        \begin{itemize}
            \item If Condition~\ref{adjcond1a} is verified then in every compatible FT-DAG $\FTCG$, one has $\parents{X_{t-\gamma}}{\FTCG} \subseteq \mathbb{Z}$ and $\descendants{X_{t-\gamma}}{\FTCG} \cap \mathbb{Z} = \emptyset$ so $\mathbb{Z}$ is a valid backdoor set and thus Equation~\ref{eq:adjustment_formula} is correct.
            \item If Condition~\ref{adjcond1b} is verified then in every compatible FT-DAG $\FTCG$, and for every backdoor path $\pi$ from $X_{t-\gamma}$ to $Y_t$ one can consider the following:
            \begin{itemize}
                \item Either the last edge of $\pi$ goes out of $Y_t$, then because $\scc{X}{\SCG} = \emptyset$ and $X\in \ancestors{Y}{\SCG}$, one knows that there exists a collider in $\pi$ which is in $\descendants{Y_t}{\SCG}\subseteq \posdescendants{X_{t-\gamma}}{\SCG}$. In this case, $\pi$ is blocked because of this collider.
                \item Either there exists $V_{t_V}$ in $\pi$ such that $t_V > t$, then there exists a collider $C_{t_C}$ with $t_C > t$. In this case, $\pi$ is blocked because of this collider.
                \item Either the last element of $\pi$ which is not in $\temp{\ecn{X}{Y}{\SCG}}{t}{t-\gamma}$ is a children of $\temp{\ecn{X}{Y}{\SCG}}{t}{t-\gamma}$ then because $\scc{X}{\SCG}=\{X\}$ one knows that $\ecn{X}{Y}{\SCG}\cap\ancestors{X}{\SCG} = \emptyset$ so there exists a collider in $\pi$ which is in $\descendants{\temp{\ecn{X}{Y}{\SCG}}{t}{t-\gamma}}{\FTCG} \subseteq \posdescendants{X_{t-\gamma}}{\SCG}$. In this case, $\pi$ is blocked because of this collider.
                \item Either the last element of $\pi$ which is not in $\temp{\ecn{X}{Y}{\SCG}}{t}{t-\gamma}$ is a parent of $\temp{\ecn{X}{Y}{\SCG}}{t}{t-\gamma}$ then one knows that this element is the middle node of either a fork or a chain in $\pi$ and that it is in $\parents{\temp{\ecn{X}{Y}{\SCG}}{t}{t-\gamma}}{\FTCG} \subseteq \mathbb{P} = \temp{\parents{\ecn{X}{Y}{\SCG}\backslash\{X\}}{\SCG}}{t}{t-\gamma-\gamma_{max}}\backslash\posdescendants{X_{t-\gamma}}{\SCG}$. In this case, $\pi$ is blocked because of this fork or chain.
                \item Note that the condition "$\cycles{X}{\SCG} = \emptyset$" guarantees the existence of an element of $\pi$ which is not in $\temp{\ecn{X}{Y}{\SCG}}{t}{t-\gamma}$.
            \end{itemize}
            \item If Condition~\ref{adjcond1c} is verified then in every compatible FT-DAG $\FTCG$, and for every backdoor path $\pi$ from $X_{t-\gamma}$ to $Y_t$ one can consider the following:
            \begin{itemize}
                \item Either the last edge of $\pi$ goes out of $Y_t$, then because $\scc{X}{\SCG} = \emptyset$ and $X\in \ancestors{Y}{\SCG}$, one knows that there exists a collider in $\pi$ which is in $\descendants{Y_t}{\SCG}\subseteq \posdescendants{X_{t-\gamma}}{\SCG}$. In this case, $\pi$ is blocked because of this collider.
                \item Either there exists $V_{t_V}$ in $\pi$ such that $t_V > t$, then there exists a collider $C_{t_C}$ with $t_C > t$. In this case, $\pi$ is blocked because of this collider.
                \item Either the last element of $\pi$ which is not in $\temp{\cn{X}{Y}{\SCG}}{t}{t-\gamma}$ is a children of $\temp{\cn{X}{Y}{\SCG}}{t}{t-\gamma}$ then because $\scc{X}{\SCG}=\{X\}$ one knows that $\cn{X}{Y}{\SCG}\cap\ancestors{X}{\SCG} = \emptyset$ so there exists a collider in $\pi$ which is in $\descendants{\temp{\cn{X}{Y}{\SCG}}{t}{t-\gamma}}{\FTCG} \subseteq \posdescendants{X_{t-\gamma}}{\SCG}$. In this case, $\pi$ is blocked because of this collider.
                \item Either the last element of $\pi$ which is not in $\temp{\cn{X}{Y}{\SCG}}{t}{t-\gamma}$ is a parent of $\temp{\cn{X}{Y}{\SCG}}{t}{t-\gamma}$ then one knows that this element is the middle node of either a fork or a chain in $\pi$ and that it is in $\parents{\temp{\cn{X}{Y}{\SCG}}{t}{t-\gamma}}{\FTCG}$. Either this node is not in $\posdescendants{X_{t-\gamma}}{\SCG}$ and $\pi$ is blocked because of this fork or chain. Either this node is in $\posdescendants{X_{t-\gamma}}{\SCG}$ in which case one should consider the last element of $\pi$ which is not in $\temp{\ecn{X}{Y}{\SCG}}{t}{t-\gamma}$: if this element is a children of $\temp{\ecn{X}{Y}{\SCG}}{t}{t-\gamma}$ then $\pi$ is blocked by a collider in $\posdescendants{X_{t-\gamma}}{\SCG}$ and if this element is a parent of $\temp{\ecn{X}{Y}{\SCG}}{t}{t-\gamma}$ then either $\pi$ is blocked by $\mathbb{Z}_2$ or we adjust on this element thus blocking $\pi$.
                \item Note that the condition "$\gamma=0$" guarantees the existence of an element of $\pi$ which is not in $\temp{\cn{X}{Y}{\SCG}}{t}{t-\gamma}$.
            \end{itemize} 
            \item If Condition~\ref{adjcond1d} is verified then a similar reasoning as for Condition~\ref{adjcond1b} can be applied. The only difference lies in the fact that if the last element of $\pi$ which is not in $\temp{\ecn{X}{Y}{\SCG}}{t}{t-\gamma}$ and is a parent of $\ecn{X}{Y}{\SCG}$ is $X_{t'}$ for $t-\gamma < t'$ then $X_{t'}\in \posdescendants{X_{t-\gamma}}{\SCG}$ and $\pi$ will not be blocked, which is why one needs to adjust on $\temp{\parents{X}{\SCG}}{t-\gamma-\gamma_{max}+1}{t}$ as well in order to block the path.
            \end{itemize}
        \item  Suppose Condition~\ref{cond2} is verified
            \begin{itemize}
            \item If Condition~\ref{adjcond2a} is verified then in every compatible FT-DAG $\FTCG$ and for every backdoor path, the first variable which is not in  $\posdescendants{X_{t-\gamma}}{\SCG}$ is in $\mathbb{Z}$ thus Equation~\ref{eq:adjustment_formula} is correct.
            \item If Condition~\ref{adjcond2b} is verified then in every compatible FT-DAG $\FTCG$, and for every backdoor path $\pi$ from $X_{t-\gamma}$ to $Y_t$ one can consider the following:
            \begin{itemize}
                \item Either the last edge of $\pi$ goes out of $Y_t$, then because $\scc{X}{\SCG} = \emptyset$ and $X\in \ancestors{Y}{\SCG}$, one knows that there exists a collider in $\pi$ which is in $\descendants{Y_t}{\SCG}\subseteq \posdescendants{X_{t-\gamma}}{\SCG}$. In this case, $\pi$ is blocked because of this collider.
                \item Either there exists $V_{t_V}$ in $\pi$ such that $t_V > t$, then there exists a collider $C_{t_C}$ with $t_C > t$. In this case, $\pi$ is blocked because of this collider.
                \item Either the last element of $\pi$ which is not in $\temp{\cn{X}{Y}{\SCG}}{t}{t-\gamma}$ is a children of $\temp{\cn{X}{Y}{\SCG}}{t}{t-\gamma}$ then because $\scc{X}{\SCG}=\{X\}$ one knows that $\cn{X}{Y}{\SCG}\cap\ancestors{X}{\SCG} = \emptyset$ so there exists a collider in $\pi$ which is in $\descendants{\temp{\cn{X}{Y}{\SCG}}{t}{t-\gamma}}{\FTCG} \subseteq \posdescendants{X_{t-\gamma}}{\SCG}$. In this case, $\pi$ is blocked because of this collider.
                \item Either the last element of $\pi$ which is not in $\temp{\cn{X}{Y}{\SCG}}{t}{t-\gamma}$ is a parent of $\temp{\cn{X}{Y}{\SCG}}{t}{t-\gamma}$ then one knows that this element is the middle node of either a fork or a chain in $\pi$ and that it is in $\parents{\temp{\cn{X}{Y}{\SCG}}{t}{t-\gamma}}{\FTCG}$. Either this node is not in $\posdescendants{X_{t-\gamma}}{\SCG}$ and $\pi$ is blocked because of this fork or chain. Either this node is in $\posdescendants{X_{t-\gamma}}{\SCG}$ in which case one should consider the last element of $\pi$ which is not in $\temp{\ecn{X}{Y}{\SCG}}{t}{t-\gamma}$: if this element is a children of $\temp{\ecn{X}{Y}{\SCG}}{t}{t-\gamma}$ then $\pi$ is blocked by a collider in $\posdescendants{X_{t-\gamma}}{\SCG}$ and if this element is a parent of $\temp{\ecn{X}{Y}{\SCG}}{t}{t-\gamma}$ then either $\pi$ is blocked by $\mathbb{Z}_2$ or we adjust on this element thus blocking $\pi$.
                \item Note that the condition "$\gamma=0$" guarantees the existence of an element of $\pi$ which is not in $\temp{\cn{X}{Y}{\SCG}}{t}{t-\gamma}$.
            \end{itemize} 
            \end{itemize}
        \item  Suppose Condition~\ref{cond3} is verified, let $\pi = \langle V^0_{t_0} \cdots V^n_{t_n}\rangle$ be a backdoor path in a compatible FT-DAG. If $\langle V^{n-1}_{t_{n-1}} \leftarrow V^n_{t_n}\rangle \subseteq \pi$ then $\pi$ contains a collider which is a descendant of $Y_t$ and so $\pi$ is blocked by this collider.
        If $\langle V^{n-1}_{t_{n-1}} \rightarrow V^n_{t_n}\rangle \subseteq \pi$ then either $V^{n-1}_{t_{n-1}} \in \mathbb{Z}$ and $\pi$ is blocked, or $V^{n-1}_{t_{n-1}} \in \mathbb{D}$ (\ie, $V^{n-1}_{t_{n-1}} = X_t$).
        In that case, one can consider the edge between $V^{n-2}_{t_{n-2}}$: either $\langle V^{n-2}_{t_{n-2}} \leftarrow V^{n-1}_{t_{n-1}} \rangle \subseteq \pi$ and thus $\pi$ is blocked by a collider which is a descendant of $X_t$, either $\langle V^{n-2}_{t_{n-2}} \rightarrow V^{n-1}_{t_{n-1}}\rangle \subseteq \pi$ then either $V^{n-2}_{t_{n-2}} \in \mathbb{Z}$ and $\pi$ is blocked, or $V^{n-2}_{t_{n-2}} \in \mathbb{D}$ (\ie, $V^{n-2}_{t_{n-2}} = Y_{t-\gamma}$).
        Notice that the path $\pi=\langle X_{t-\gamma} \leftarrow Y_{t-\gamma} \rightarrow X_t \rightarrow Y_t \rangle$ cannot exist because of causal stationarity, the edges $X_{t-\gamma} \leftarrow Y_{t-\gamma}$ and $X_t \rightarrow Y_t$ cannot both exist in the same FT-DAG.
        If $\langle V^{n-3}_{t_{n-3}} \leftarrow V^{n-2}_{t_{n-2}} \rangle \subseteq \pi$ then $\pi$ is blocked by a collider which is a descendant of $Y_{t-\gamma}$. Otherwise, if$\langle V^{n-3}_{t_{n-3}} \rightarrow V^{n-2}_{t_{n-2}} \rangle \subseteq \pi$ and $\mathbb{P}^{all}_Y \subseteq \mathbb{Z}$ then $V^{n-3}_{t_{n-3}} \in \mathbb{P}_Y \cup \mathbb{P}^{all}_Y \subseteq \mathbb{Z}$ blocks $\pi$.
        Lastly, if $\mathbb{P}^{all}_X \subseteq \mathbb{Z}$ then $V^1_{t_1}\in\mathbb{P}_X \cup \mathbb{P}^{all}_X \subseteq \mathbb{Z}$  blocks $\pi$ except if $V^1_{t_1} \in \mathbb{D}$ which is only possible if $V^1_{t_1} = Y_{t-\gamma}$ which we have shown before to be impossible due to causal stationarity.
    \end{itemize}
\end{proof}

\propositionOptimalIsGivenBySCGBackDoor*
\begin{proof}
    Let $\mathbb{D}=\posdescendants{X_{t-\gamma}}{\SCG}$.
    \begin{itemize}
        \item If Condition~\ref{cond1} is verified and $\cycles{X}{\SCG} = \emptyset$ and $\gamma>0$, then let $\mathbb{P} = \temp{\parents{\ecn{X}{Y}{\SCG}}{\SCG}}{t}{t-\gamma-\gamma_{max}}$ and $\qopt{X_{t-\gamma}}{Y_t}{\SCG}=\mathbb{P}\backslash \mathbb{D}$. $\qopt{X_{t-\gamma}}{Y_t}{\SCG}$ verifies item~\ref{adjcond1b} of Definition~\ref{def:backdoor}.
        \item If Condition~\ref{cond1} is verified and $\gamma=0$, then let $\mathbb{P} = \temp{\parents{\cn{X}{Y}{\SCG}}{\SCG}}{t}{t-\gamma-\gamma_{max}} \cup \temp{\parents{\ecnbd{X}{Y}{\SCG}{\emptyset}}{\SCG}}{t}{t-\gamma-\gamma_{max}}$ and $\qopt{X_{t-\gamma}}{Y_t}{\SCG}=\mathbb{P}\backslash \mathbb{D}$. $\qopt{X_{t-\gamma}}{Y_t}{\SCG}$ verifies item~\ref{adjcond1c} of Definition~\ref{def:backdoor}.
        \item If Condition~\ref{cond1} is verified and $\cycles{X}{\SCG} \neq \emptyset$ and $\gamma>0$, then let $\mathbb{P} = \temp{\parents{X}{\SCG}}{t}{t-\gamma-\gamma_{max}+1} \cup \temp{\parents{\ecn{X}{Y}{\SCG}}{\SCG}}{t}{t-\gamma-\gamma_{max}}$ and $\qopt{X_{t-\gamma}}{Y_t}{\SCG}=\mathbb{P}\backslash \mathbb{D}$. $\qopt{X_{t-\gamma}}{Y_t}{\SCG}$ verifies item~\ref{adjcond1d} of Definition~\ref{def:backdoor}.
        \item If Condition~\ref{cond2} is verified, then let $\mathbb{P} = \temp{\parents{\cn{X}{Y}{\SCG}}{\SCG}}{t}{t-\gamma-\gamma_{max}} \cup \temp{\parents{\ecnbd{X}{Y}{\SCG}{\emptyset}}{\SCG}}{t}{t-\gamma-\gamma_{max}}$ and $\qopt{X_{t-\gamma}}{Y_t}{\SCG}=\mathbb{P}\backslash \mathbb{D}$. $\qopt{X_{t-\gamma}}{Y_t}{\SCG}$ verifies item~\ref{adjcond2b} of Definition~\ref{def:backdoor}.
        \item If Condition~\ref{cond3} is verified, then let $\mathbb{P} = \temp{\parents{Y}{\SCG}}{t}{t-\gamma-\gamma_{max}} \cup \temp{\parents{X}{\SCG}}{t}{t-\gamma_{max}}$ and $\qopt{X_{t-\gamma}}{Y_t}{\SCG}=\mathbb{P}\backslash \mathbb{D}$. $\qopt{X_{t-\gamma}}{Y_t}{\SCG}$ verifies Definition~\ref{def:backdoor}.
    \end{itemize}
\end{proof}

\propositionOptimalInOneFTDAG*
\begin{proof}
    Build $\FTCG$ the following way:
    \begin{itemize}
        \item $\forall \langle U\rightarrow V \rangle \in \SCG$ and $\forall t-\gamma_{max} \leq t_U < t$ add the edge $\langle U_{t_U} \rightarrow V_t \rangle$ to $\FTCG$.
        \item For all directed path from $X$ to $Y$ $\pi=\langle V^0 \rightarrow \cdots \rightarrow V^n\rangle$ and $\forall 0<i<n$ iteratively add the edges $\forall t,~\langle V^i_t \rightarrow V^{i+1}_t\rangle$ to $\FTCG$ if it does not create a cycle.
        \item $\forall \langle U_t \rightarrow V_t \rangle \in \FTCG$ and $\forall P \in \parents{V}{\SCG}$ add the edges $\forall t,~\langle P_t \rightarrow V_t\rangle$ to $\FTCG$ if it does not create a cycle.
        \item Repeat this last item until no additional edge can be added.
    \end{itemize}
    With this construction it is clear that $\qopt{X_{t-\gamma}}{Y_t}{\SCG} = \opt{X_{t-\gamma}}{Y_t}{\FTCG}$ as every potential parent of every extended causal node (or causal node if $\gamma=0$) in the SCG is a parent of a causal node in the FT-DAG. Thus the only elements of $\qopt{X_{t-\gamma}}{Y_t}{\SCG}$ which are not in $\opt{X_{t-\gamma}}{Y_t}{\FTCG}$ are the variables in $\temp{\parents{\ecnbd{X}{Y}{\SCG}{\emptyset}}{\SCG}}{t}{t-\gamma-\gamma_{max}}$. However, $\ecnbd{X}{Y}{\SCG}{\emptyset}$ is always empty except in cases such as the one of Figure~\ref{fig:ex_qopt_diff_opt} in which there does not exist any compatible FT-DAG $\FTCG$ such that $\opt{X_{t-\gamma}}{Y_t}{\FTCG}$ is a valid adjustment set in every compatible FT-DAG.
\end{proof}

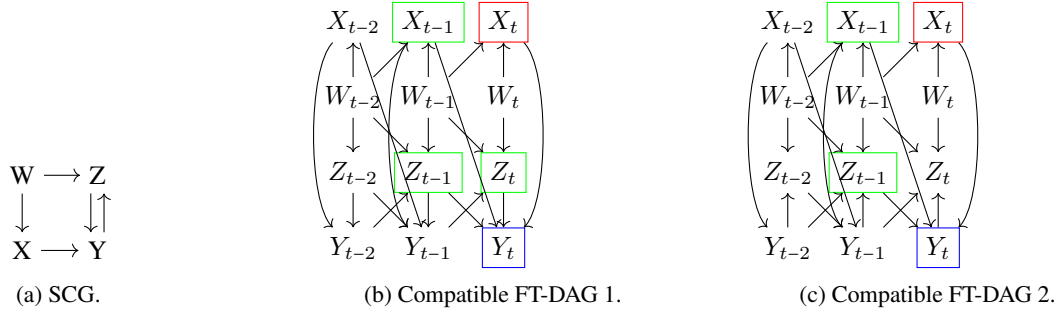
\begin{figure}[t]
\centering
\begin{subfigure}[t]{.3\linewidth}
\centering
\begin{tikzpicture}
    \node (X) {X};
    \node (Y) [right of=X] {Y};
    \node (W) [above of=X] {W};
    \node (Z) [above of=Y] {Z};
    
    \draw [->] (X) -- (Y);
    \draw [->] (W) -- (X);
    \draw [->] (W) -- (Z);
    \begin{scope}[transform canvas={xshift=-.25em}]
    \draw [->] (Z) -- (Y);
    \end{scope}
    \begin{scope}[transform canvas={xshift=.25em}]
    \draw [<-] (Z) -- (Y);
    \end{scope}
\end{tikzpicture}
\caption{SCG.}
\label{fig:ex_qopt_diff_opt:SCG}
\end{subfigure}
\hfill
\begin{subfigure}[t]{.3\linewidth}
\begin{tikzpicture}
  \node (X2) {$X_{t-2}$};
  \node [draw = green] (X1) [right of= X2] {$X_{t-1}$};
  \node [draw=red] (X0) [right of= X1] {$X_{t}$};

  \node (W2) [below of= X2] {$W_{t-2}$};
  \node (W1) [right of= W2] {$W_{t-1}$};
  \node (W0) [right of= W1] {$W_{t}$};

  \node (Z2) [below of= W2] {$Z_{t-2}$};
  \node [draw = green] (Z1) [right of= Z2] {$Z_{t-1}$};
  \node [draw = green] (Z0) [right of= Z1] {$Z_{t}$};
  
  \node (Y2) [below of= Z2] {$Y_{t-2}$};
  \node (Y1) [right of= Y2] {$Y_{t-1}$};
  \node [draw = blue] (Y0) [right of= Y1] {$Y_{t}$};
    
    \draw [->] (W2) -- (X2);
    \draw [->] (W2) -- (X1);
    \draw [->] (W1) -- (X1);
    \draw [->] (W1) -- (X0);
    \draw [->] (W0) -- (X0);
    
    \draw [->] (W2) -- (Z2);
    \draw [->] (W2) -- (Z1);
    \draw [->] (W1) -- (Z1);
    \draw [->] (W1) -- (Z0);    
    \draw [->] (W0) -- (Z0);

    \draw [->, looseness=.5, out=-90-45, in=90+45] (X2) to (Y2);
    \draw [->] (X2) -- (Y1);
    \draw [->, looseness=.5, out=-90-45, in=90+45] (X1) to (Y1);
    \draw [->] (X1) -- (Y0);    
    \draw [->, looseness=.5, out=-90+45, in=90-45] (X0) to (Y0);

    \draw [->] (Z2) -- (Y1);
    \draw [->] (Z1) -- (Y0);
    \draw [->] (Y2) -- (Z1);
    \draw [->] (Y1) -- (Z0);

    \draw [->] (Z2) -- (Y2);
    \draw [->] (Z1) -- (Y1);
    \draw [->] (Z0) -- (Y0);
\end{tikzpicture}
\caption{Compatible FT-DAG 1.}
\label{fig:ex_qopt_diff_opt:DAG1}
\end{subfigure}
\hfill
\begin{subfigure}[t]{.3\linewidth}
\begin{tikzpicture}
  \node (X2) {$X_{t-2}$};
  \node [draw = green] (X1) [right of= X2] {$X_{t-1}$};
  \node [draw=red] (X0) [right of= X1] {$X_{t}$};

  \node (W2) [below of= X2] {$W_{t-2}$};
  \node (W1) [right of= W2] {$W_{t-1}$};
  \node (W0) [right of= W1] {$W_{t}$};

  \node (Z2) [below of= W2] {$Z_{t-2}$};
  \node [draw = green] (Z1) [right of= Z2] {$Z_{t-1}$};
  \node (Z0) [right of= Z1] {$Z_{t}$};
  
  \node (Y2) [below of= Z2] {$Y_{t-2}$};
  \node (Y1) [right of= Y2] {$Y_{t-1}$};
  \node [draw = blue] (Y0) [right of= Y1] {$Y_{t}$};
    
    \draw [->] (W2) -- (X2);
    \draw [->] (W2) -- (X1);
    \draw [->] (W1) -- (X1);
    \draw [->] (W1) -- (X0);
    \draw [->] (W0) -- (X0);
    
    \draw [->] (W2) -- (Z2);
    \draw [->] (W2) -- (Z1);
    \draw [->] (W1) -- (Z1);
    \draw [->] (W1) -- (Z0);    
    \draw [->] (W0) -- (Z0);

    \draw [->, looseness=.5, out=-90-45, in=90+45] (X2) to (Y2);
    \draw [->] (X2) -- (Y1);
    \draw [->, looseness=.5, out=-90-45, in=90+45] (X1) to (Y1);
    \draw [->] (X1) -- (Y0);    
    \draw [->, looseness=.5, out=-90+45, in=90-45] (X0) to (Y0);

    \draw [->] (Z2) -- (Y1);
    \draw [->] (Z1) -- (Y0);
    \draw [->] (Y2) -- (Z1);
    \draw [->] (Y1) -- (Z0);

    \draw [->] (Y2) -- (Z2);
    \draw [->] (Y1) -- (Z1);
    \draw [->] (Y0) -- (Z0);
\end{tikzpicture}
\caption{Compatible FT-DAG 2.}
\label{fig:ex_qopt_diff_opt:DAG2}
\end{subfigure}
\caption{Example of a SCG such that for every FT-DAG $\FTCG$ there exists another FT-DAG $\FTCG'$ in which $\opt{X_{t-\gamma}}{Y_t}{\FTCG}$ is not a valid adjustment set. In red is the treatment variable, in blue the outcome variable and in green is the optimal set. Notice that in the first FT-DAG the backdoor path $\langle X_t \leftarrow W_t \rightarrow Z_t \rightarrow Y_t\rangle$ would not be blocked if one did not adjust on $Z_t$, and that in the second FT-DAG $Z_t$ is a descendant of $Y_t$ so one cannot adjust on it without introducing selection bias. Thus, each optimal adjustment set is invalid in the other compatible FT-DAG.}
\label{fig:ex_qopt_diff_opt}
\end{figure}

\propositionIncludeAllOptimals*
\begin{proof}
    Let $\FTCG$ be a FT-DAG.  If $\gamma= 0$ then $\cn{X_{t-\gamma}}{Y_t}{\FTCG} = \temp{\cn{X_{t-\gamma}}{Y_t}{\SCG}}{t}{t-\gamma}$ and if $\gamma> 0$ then $\cn{X_{t-\gamma}}{Y_t}{\FTCG} = \temp{\ecn{X_{t-\gamma}}{Y_t}{\SCG}}{t}{t-\gamma}$. Therefore, $\parents{\cn{X_{t-\gamma}}{Y_t}{\FTCG}}{\FTCG} \subseteq \temp{\parents{\cn{X_{t-\gamma}}{Y_t}{\SCG}}{\SCG}}{t}{t-\gamma-\gamma_{max}}$ when $\gamma =0$ and $\parents{\cn{X_{t-\gamma}}{Y_t}{\FTCG}}{\FTCG} \subseteq \temp{\parents{\ecn{X_{t-\gamma}}{Y_t}{\SCG}}{\SCG}}{t}{t-\gamma-\gamma_{max}}$ when $\gamma>0$. If Conditions~\ref{cond1} or~\ref{cond2} are satisfied, then $\temp{\parents{\ecn{X_{t-\gamma}}{Y_t}{\SCG}}{\SCG}}{t}{t-\gamma-\gamma_{max}} \subseteq \mathbb{P}$ where $\mathbb{P}$ is defined in Definition~\ref{def:q-opt-scg}. In these cases, we have $\opt{X_{t-\gamma}}{Y_t}{\FTCG} \backslash \posdescendants{X_{t-\gamma}}{\SCG}\subseteq\qopt{X_{t-\gamma}}{Y_{t}}{\SCG}$.
    Lastly, if Condition~\ref{cond3} is satisfied then $\cn{X_{t-\gamma}}{Y_t}{\FTCG} = \{Y_{t-1},X_t\}$ so it is clear that $\opt{X_{t-\gamma}}{Y_t}{\FTCG} \backslash \posdescendants{X_{t-\gamma}}{\SCG}\subseteq\qopt{X_{t-\gamma}}{Y_{t}}{\SCG}$.
\end{proof}

\corollaryUnionAllOptimals*
\begin{proof}
    It follows from Propositions~\ref{prop:OptimalInOneFTDAG} and \ref{prop:IncludeAllOptimals}.
\end{proof}

\corollaryQoptVsAs*
\begin{proof}
The proof follows directly from \cite[Theorem 1]{Henckel_2022}.

Suppose that $T=\mathbb{A}^2 \backslash \qopt{X_{t-\gamma}}{Y_t}{\mathcal{G}^s}$ and  $S = \qopt{X_{t-\gamma}}{Y_t}{\mathcal{G}^s} \backslash \mathbb{A}^2.$
By definition, $T$ consists of elements in $\mathbb{A}^2$ that are not in $\qopt{X_{t-\gamma}}{Y_t}{\mathcal{G}^s}$. Since $\qopt{X_{t-\gamma}}{Y_t}{\mathcal{G}^s}$ includes all parents of $Y_t$ (not descendants of $X_{t-\gamma}$) and all parents of causal nodes from $X_{t-\gamma}$ to $Y_t$ and does not include any descendant of $Y_t$, conditioning on $\qopt{X_{t-\gamma}}{Y_t}{\mathcal{G}^s}$ and $X_{t-\gamma}$ d-separates $Y_t$ from any variable except causal nodes from $X_{t-\gamma}$ to $Y_t$ and descendants of $Y_t$. Fortunately, $\mathbb{A}^2$ is a valid backdoor set so it does not contain any such set and $T$ does not either.
Therefore $Y_t$ is d-separated of $T$ given $\qopt{X_{t-\gamma}}{Y_t}{\mathcal{G}^s}$ and $X_{t-\gamma}$. 
In addition, $X_{t-\gamma}$ is trivially d-separated of $S$ given $\mathbb{A}^2$ as $\mathbb{A}^2$ does not contain any descendant of $X_{t-\gamma}$. 
Thus by \cite[Theorem 1]{Henckel_2022}, $\var{\Pr}{\estimator{X_{t-\gamma}}{Y_t}{\qopt{X_{t-\gamma}}{Y_t}{\mathcal{G}^s}}} \le \var{\Pr}{\estimator{X_{t-\gamma}}{Y_t}{\mathbb{A}^2}}$.
Same reasoning can be made to conclude by \cite[Theorem 1]{Henckel_2022} that $\var{\Pr}{\estimator{X_{t-\gamma}}{Y_t}{\qopt{X_{t-\gamma}}{Y_t}{\mathcal{G}^s}}} \le \var{\Pr}{\estimator{X_{t-\gamma}}{Y_t}{\mathbb{A}^1}}$.
\end{proof}

\end{document}